\pdfoutput=1
\RequirePackage{ifpdf}
\ifpdf % We are running pdfTeX in pdf mode
\documentclass[pdftex]{sigma}
\else
\documentclass{sigma}
\fi

\usepackage[all]{xy}

\newtheorem{theor}{Theorem}[section]
\newtheorem{predl}[theor]{Proposition}
\newtheorem{lem}[theor]{Lemma}
\newtheorem{cor}[theor]{Corollary}
{\theoremstyle{definition}
\newtheorem{defi}[theor]{Definition}
\newtheorem{ex}[theor]{Example} }

\numberwithin{equation}{section}

\begin{document}

\allowdisplaybreaks

\newcommand{\arXivNumber}{1412.8116}

\renewcommand{\PaperNumber}{084}

\FirstPageHeading

\ShortArticleName{Bruhat Order in the Full Symmetric $\mathfrak{sl}_n$ Toda Lattice on Partial Flag Space}

\ArticleName{Bruhat Order in the Full Symmetric $\boldsymbol{\mathfrak{sl}_n}$ Toda Lattice\\ on Partial Flag Space}

\Author{Yury B.~CHERNYAKOV~$^{\dag^1\dag^2}$, Georgy I.~SHARYGIN~$^{\dag^1\dag^2\dag^3}$ and Alexander S.~SORIN~$^{\dag^2\dag^4\dag^5}$}

\AuthorNameForHeading{Yu.B.~Chernyakov, G.I.~Sharygin and A.S.~Sorin}

\Address{$^{\dag^1}$~Institute for Theoretical and Experimental Physics, 25 Bolshaya Cheremushkinskaya,\\
\hphantom{$^{\dag^1}$}~117218, Moscow, Russia}
\EmailDD{\href{mailto:chernyakov@itep.ru}{chernyakov@itep.ru}, \href{mailto:sharygin@itep.ru}{sharygin@itep.ru}}

\Address{$^{\dag^2}$~Joint Institute for Nuclear Research, Bogoliubov Laboratory of Theoretical Physics,\\
\hphantom{$^{\dag^2}$}~141980, Dubna, Moscow region, Russia}
\EmailDD{\href{mailto:sorin@theor.jinr.ru}{sorin@theor.jinr.ru}}

\Address{$^{\dag^3}$~Lomonosov Moscow State University, Faculty of Mechanics and Mathematics,\\
\hphantom{$^{\dag^3}$}~GSP-1, 1~Leninskiye Gory, Main Building, 119991, Moscow, Russia}

\Address{$^{\dag^4}$~National Research Nuclear University MEPhI (Moscow Engineering Physics Institute),\\
\hphantom{$^{\dag^4}$}~31 Kashirskoye Shosse, 115409 Moscow, Russia}

\Address{$^{\dag^5}$~Dubna International University, 141980, Dubna, Moscow region, Russia}

\ArticleDates{Received February 15, 2016, in f\/inal form August 10, 2016; Published online August 20, 2016}

\Abstract{In our previous paper [\textit{Comm. Math. Phys.} \textbf{330} (2014), 367--399] we described the asymptotic behaviour of trajectories of the full symmetric $\mathfrak{sl}_n$ Toda lattice in the case of distinct eigenvalues of the Lax matrix. It turned out that it is completely determined by the Bruhat order on the permutation group. In the present paper we extend this result to the case when some eigenvalues of the Lax matrix coincide. In that case the trajectories are described in terms of the projection to a partial f\/lag space where the induced dynamical system verif\/ies the same properties as before: we show that when $t\to\pm\infty$ the trajectories of the induced dynamical system converge to a f\/inite set of points in the partial f\/lag space indexed by the Schubert cells so that any two points of this set are connected by a trajectory if and only if the corresponding cells are adjacent. This relation can be explained in terms of the Bruhat order on multiset permutations.}


\Classification{06A06; 37D15; 37J35}

\section{Introduction}
\subsection{Toda system}

The present paper is devoted to the study of the full symmetric $\mathfrak{sl}_n$ Toda lattice which can be considered as a straightforward generalization of the non-periodic Toda lattice. Let us brief\/ly remind that the non-periodic Toda lattice (Toda chain) is the dynamical system of $n$ particles on a straight line with interactions between neighbours. This system was f\/irst considered in \cite{T1, T2}; in paper \cite{H} $n$ functionally independent integrals of the motion were found. The involution of the integrals was proved in \cite{F1, F2}.

\newpage

The non-periodic Toda lattice has the Lax representation and the Lax operator matrix has the following form in Flaschka's variables:
\begin{gather*} %{Lax}
 L = \left(
\begin{matrix}
 b_{1} & a_{1} & 0 & \cdots &0\\
 a_{1} & b_{2} & a_{2} &\cdots & 0\\
 0 & \cdots & \cdots & \cdots& 0\\
 0 & \cdots & a_{n-2} & b_{n-1} &a_{n-1}\\
 0 & 0 & \cdots & a_{n-1} & b_{n}
\end{matrix}
\right)
\end{gather*}
One can show that the Hamilton equations are equivalent to the following matrix equation
\begin{gather}\label{LM}
\dot{L} = [B,L],
\end{gather}
where $B$ is
\begin{gather*} %{Lax2}
B = \left(
\begin{matrix}
 0 & -a_{1} & 0 & \cdots &0\\
 a_{1} & 0 & -a_{2} & \cdots & 0\\
 0 & \cdots & \cdots & \cdots & 0\\
 0 & \cdots & a_{n-2} & 0 & -a_{n-1}\\
 0 & 0 & \cdots & a_{n-1} & 0
\end{matrix}
\right).
\end{gather*}
Equation (\ref{LM}) is the compatibility condition of the system
\begin{gather*} %{compat-Lax}
L \Psi = \Psi \Lambda,\qquad
\frac{\partial}{\partial t} \Psi = B \Psi,
\end{gather*}
where $\Psi \in {\rm SO}(n,\mathbb{R})$ and $\Lambda$ is the eigenvalue matrix of the Lax operator.

One can show that the non-periodic Toda lattice can be treated as a dynamical system on the orbits of the coadjoint action of the Borel subgroup $B^+_n$ of ${\rm SL}(n,\mathbb R)$ (equal to the group of upper triangular matrices with the determinant equal to~$1$) see \cite{Ad,A, K1, S}. It is also possible to give an alternative description of the phase space of this dynamical system if we identify the algebra $\mathfrak{sl}_n$ with its dual using the Killing form on $\mathfrak{sl}_n$. In this way one obtains generalizations of the classical Toda lattice (tri-diagonal Toda chain). Namely, use the following identif\/ications:
\begin{gather*}%{Decomp-2}
\mathfrak{sl}_n=\mathfrak{so}_n \oplus \mathfrak{b}^+_n,\qquad
\mathfrak{sl}^{\ast}_{n} = (\mathfrak{b}^{+}_n)^{\ast} \oplus (\mathfrak{so}_n)^{\ast} \cong \operatorname{Symm}_n \oplus \mathfrak{n}_n^{+},\\
(\mathfrak{b}^{+}_n)^{\ast} \cong (\mathfrak{so}_n)^{\perp} = \operatorname{Symm}_n, \qquad (\mathfrak{so}_n)^{\ast} \cong (\mathfrak{b}^{+}_n)^{\perp} = \mathfrak{n}_n^{+},
\end{gather*}
where $\mathfrak{b}^+_n$ is the algebra of the upper triangular matrices and $\mathfrak{n}_n^{+}$ is the algebra of the strictly upper triangular matrices. As one can see this identif\/ication maps the space of symmetric matrices into the dual space of the Lie algebra of the Borel subgroup: $\operatorname{Symm}_n\cong(\mathfrak{b}_n^+)^*$ and hence we can introduce Kirillov--Kostant symplectic structure on $\operatorname{Symm}_n$ pulling it back from $(\mathfrak{b}_n^+)^*$; it is by restriction of this pullback that one obtains the symplectic structure on the tridiagonal symmetric matrices used in the Toda system.

Based on this approach one can get further generalizations of the non-periodic Toda lattice just by plugging in other Cartan pairs in this construction. In particular in this way one obtains the full symmetric $\mathfrak{sl}_n$ Toda lattice (elsewhere FS Toda lattice). In contrast to the Lax matrix of the non-periodic Toda lattice the Lax matrix of the FS Toda lattice is not a tri-diagonal symmetric matrix but an arbitrary symmetric matrix:
\begin{gather*} %{Lax2}
L = \left(
\begin{matrix}
 a_{11} & a_{12} & a_{13} & \cdots & a_{1n}\\
 a_{12} & a_{22} & a_{23} & \cdots & a_{2n}\\
 \cdots & \cdots & \cdots & \cdots & \cdots\\
 a_{1n} & a_{2n} & a_{3n} & \cdots & a_{nn}
\end{matrix}
\right).
\end{gather*}
As one knows every symmetric matrix can be diagonalized in an orthogonal basis, i.e., the Lax matrix $L$ can be decomposed as follows
\begin{gather*} %{compat-Lax}
L = \Psi \Lambda \Psi^{-1},
\end{gather*}
where $\Psi$ is an orthogonal matrix. It turns out, that the FS Toda lattice is also integrable (we refer the reader to \cite{A, CS,DLNT, EFS, FS} for details). Moser in~\cite{Mos} showed that as $t \rightarrow -\infty$ the Lax operator of the usual tri-diagonal Toda lattice converges to the diagonal matrix with eigenvalues put in increasing order, and when $t \rightarrow +\infty$ it converges to the diagonal matrix with decreasing eigenvalues. This property was further studied in~\cite{CSS,DNT, FS,KM,KW}.

\subsection{Outline of the paper}

The present paper deals with the FS Toda lattice in an arbitrary dimension~$n$. It is a natural continuation of the previous one~\cite{CSS} in which we study the behavior of the Lax operator $L$ as $t \rightarrow \pm\infty$; this question is again in the center of our attention here. The dif\/ference is that in \cite{CSS} we treated the Lax matrices $L$ with $n$ distinct eigenvalues and now we let the eigenvalues of~$L$ coincide. In the previous case the Bruhat order on the symmetric group $S_n$ played a crucial r\^ole. On the other hand, as one knows there exists a Bruhat order on the permutations of multisets induced in a natural way from the Bruhat order on the symmetric group (one can regard this phenomenon as a combinatoric manifestation of the fact that the Schubert cells in the full f\/lag space project into the Schubert cells in the partial f\/lag manifolds). So it is natural to ask if this (induced) order has something to do with the asymptotic behavior of the Lax matrix with coinciding eigenvalues.

We show that one can regard the restriction of the FS Toda lattice on the space of symmetric matrices with multiple eigenvalues as a gradient system on a partial f\/lag space so that the set of singular points of the gradient vector f\/ield is naturally identif\/ied with the permutations of multisets (see Section~\ref{s:tod}, Proposition~\ref{morsepartflag}) and we show that the gradient system at hand verif\/ies the same properties as the system with distinct eigenvalues. That is we prove that for $t\to\pm\infty$ the trajectories of the FS Toda lattice converge to a f\/inite set of singular points in the partial f\/lag space indexed by the Schubert cells so that any two points of this set are connected by a~trajectory if and only if the corresponding cells are adjacent (see Theorem~\ref{prop:traject1}).

The paper is organized as follows: in Section~\ref{s:bru} we give a list of facts from the geometry of partial f\/lag spaces, describe the Bruhat order and introduce Schubert cells; all this is used in the rest of the paper. In Section~\ref{s:tod} we consider the FS Toda lattice and describe how it induces a~gradient f\/low on partial f\/lag manifolds so that some elementary facts from the Morse theory are applicable. Finally, in Section~\ref{ss:fin:ex} we give two simple explicit examples that illustrate our theorem and prove the main results of this paper (see Section~\ref{ss:fin:gc}).

\subsection{Notation and assumptions}\label{ss:intro:not}
In what follows, unless otherwise stated, all manifolds will be assumed smooth and compact (without boundary), all vector spaces are assumed to be real and f\/inite-dimensional.

We shall also consider the full symmetric Toda system in a generic dimension~$n$, so we let~$L$ denote the real symmetric $n\times n$ Lax matrix of the system and $\Lambda$ the diagonal matrix of eigenvalues of~$L$. As $L$ is symmetric, its eigenvalues are real; so we assume that they are ordered naturally in $\Lambda$.

We shall use the notation ${\rm O}(n, \mathbb R)$ (respectively ${\rm SO}(n, \mathbb R)$) for the group of $n$-dimensional orthogonal matrices (respectively the group of orthogonal matrices with positive determinant), and $\mathfrak{so}(n)$ will denote its Lie algebra that is the space of real antisymmetric $n\times n$-matrices. Similarly ${\rm SL}(n, \mathbb R)$ will denote the $n$-dimensional special linear group over real numbers and $B_n^+\subset {\rm SL}(n, \mathbb R)$ (respectively $B_n^-$) the Borel group of upper (respectively lower) triangular matrices with unit determinant.

\section{Bruhat order and Schubert cells}\label{s:bru}
The notion of the Bruhat order, Schubert cells and their generalizations have long been crucial instruments in the research of geometry of Lie groups and homogenous spaces. In what follows we give a brief introduction to the subject. In most part of this section we draw on the exposition from classical books, see~\cite{F} and references therein.

\subsection{Flag spaces, Grassmanians and their generalizations}\label{ss:bru:fs&gr}
First we recall some def\/initions:
\begin{defi}\label{df:flags}
Let $I=(i_1,i_2,\dots,i_k)$ be a set of positive integers, $i_1+i_2+\dots+i_k=n$. Then the real partial f\/lag space ${\rm Fl}_{i_1,i_2,\dots,i_k}(\mathbb R)$ is the set of all sequences of vector subspaces
\begin{gather*}
\big\{0=V_0\subset V_1\subset V_2\subset V_3\subset \dots\subset V_k=\mathbb R^n\big\}
\end{gather*}
such that $\dim V_l=i_1+i_2+\dots+i_l$.
\end{defi}

Let us make a few remarks before we proceed. First of all, similar def\/initions work literally for vector spaces over arbitrary f\/ield $\Bbbk$ giving us the f\/lag spaces ${\rm Fl}_{i_1,\dots,i_k}(\Bbbk)$. These spaces are easy to def\/ine on the level of sets; however, their topology depends a lot on $\Bbbk$. Below we shall work only with $\Bbbk=\mathbb{R}$ or $\Bbbk=\mathbb{C}$ with a usual topology. Second, an important particular case of this construction is $k=n$, $i_1=i_2=\dots=i_n=1$, which is called the \textit{full flag manifold} and is denoted by ${\rm Fl}_n(\mathbb R)$ (or in a general case ${\rm Fl}_n(\Bbbk)$). Another important particular case of the f\/lag space is the case $k=1$, i.e., when there is only one subspace chosen inside $\mathbb R^n$. In this case the f\/lag manifold is called \textit{the Grassman space} or just \textit{Grassmannian} and is denoted by ${\rm Gr}_{d,n}(\mathbb R)$, where $d=i_1$ is the dimension of the subspaces we choose.

An important property of f\/lag spaces is that they are homogeneous spaces, the quotient spaces of the groups of linear transformations of $\mathbb R^n$ by their subgroups. The isomorphism is induced by the choice of a compatible basis in the subspaces $V_i$. In particular,
\begin{gather*}
{\rm Fl}_n(\mathbb R)={\rm SL}(n,\mathbb R)/B_n^+,\qquad {\rm Gr}_{d,n}(\mathbb R)={\rm SL}(n,\mathbb R)/P,
\end{gather*}
for a \textit{maximal parabolic subgroup} $P\subset {\rm SL}(n,\mathbb R)$. Similarly, the general partial f\/lag spaces are isomorphic to the quotient spaces of ${\rm SL}(n,\mathbb R)$ by a suitable parabolic subgroup. On the other hand, if we choose an orthogonal basis of the corresponding subspaces so that the orientation on~$\mathbb R^n$ would match with the given one, we obtain homeomorphisms with the quotients of a~special orthogonal group. For instance
\begin{gather*}
{\rm Fl}_n(\mathbb R)={\rm O}(n;\mathbb R)/T_n={\rm SO}(n, \mathbb R)/T^+_n,\\
{\rm Gr}_{d,n}(\mathbb R)={\rm O}(n;\mathbb R)/{\rm O}(d;\mathbb R)\times {\rm O}(n-d;\mathbb R)\\
\hphantom{{\rm Gr}_{d,n}(\mathbb R)}{}
={\rm SO}(n;\mathbb R)/{\rm SO}(n,\mathbb R)\bigcap\left({\rm O}(d;\mathbb R)\times {\rm O}(n-d;\mathbb R)\right).
\end{gather*}
Here we have denoted by $T_n$ the group of diagonal matrices with eigenavalues equal to $\pm1$, and~$T^+_n$ the intersection ${\rm SO}(n,\mathbb R)\cap T_n$. In both cases we see that the groups ${\rm SL}(n,\mathbb R)$ and ${\rm SO}(n,\mathbb R)$ act transitively on the f\/lag spaces.

The important r\^ole of the f\/lag spaces in our investigation follows from the next observation:
\begin{predl}
{The space of real symmetric $n \times n$ matrices with a fixed set of eigenvalues $\lambda_1,\dots,\lambda_k$ with multiplicities $i_1,\dots,i_k$ can be identified with the partial flag manifold ${\rm Fl}_{i_1,\dots,i_k}$.}
\end{predl}
Before we \textit{prove} this in full generality, let us consider the simplest case: $n=3$. Let $\alpha<\beta$ be real numbers. Then the set of all symmetric $3 \times 3$ matrices with eigenvalues $\alpha$, $\alpha$, $\beta$ is equal to the orbit of the diagonal matrix $\operatorname{diag}(\alpha,\alpha,\beta)$ under conjugations by the elements of ${\rm SO}(3, \mathbb{R})$. This action is not free, the stabilizer of $\operatorname{diag}(\alpha,\alpha,\beta)$ being equal to the subgroup $\widetilde{{\rm SO}}(2, \mathbb{R}) \subset {\rm SO}(3, \mathbb{R})$ of matrices $\Psi$ that have one of the following forms:
\begin{gather*}
\Psi=\begin{pmatrix}\cos{t} & \sin{t} & 0\\ -\sin{t} & \cos{t} & 0\\ 0 & 0 & 1\end{pmatrix}\qquad \text{or}\qquad
\Psi=\begin{pmatrix}\cos t & -\sin t & 0\\ -\sin t & -\cos t & 0\\ 0 & 0 & -1\end{pmatrix},
\end{gather*}
So the orbit is equal to the quotient space of the sphere ${\rm SO}(3, \mathbb{R})/{\rm SO}(2, \mathbb{R})=S^2$ by the antipodal action of $\mathbb Z/2\mathbb Z$; to see this observe that the action is induced from the conjugation by the matrix
\begin{gather*}
A=\begin{pmatrix}
1 & 0 & 0\\
0 &-1 & 0\\
0 & 0 &-1
\end{pmatrix}.
\end{gather*}
In fact, the subgroup $\widetilde{{\rm SO}}(2, \mathbb{R}) \subset {\rm SO}(3, \mathbb{R})$ is generated by $A$ and the subgroup of rotations around the $Oz$ axis. Since the homeomorphism ${\rm SO}(3, \mathbb{R})/{\rm SO}(2, \mathbb{R})=S^2$ is given by the image of the point $(0,0,1)$ under the action of ${\rm SO}(3, \mathbb{R})$, we see that the action of $A$ on $S^2$ is given by the antipodal map. So the quotient space that we need is $S^2/(\mathbb Z/2\mathbb Z)=\mathbb RP^2$. As one can see this construction amounts to sending each matrix $L$ from this set to the straight line spanned by the eigenvectors corresponding to $\beta$. Alternatively, one can consider~$\mathbb RP^2$ in the dual sense as the space of all 2-dimensional subspaces in $\mathbb{R}^{3}$, and identify~$L$ with the eigenspace corresponding to~$\alpha$.

More generally, assume that the diagonal matrix $\Lambda$ has $k<n$ coinciding eigenvalues, and all the other eigenvalues of $\Lambda$ are distinct. Without loss of generality we can think that $\lambda_1=\lambda_2=\dots=\lambda_k$. Then reasoning just as before we can identify the set of symmetric matrices with such eigenvalues with the quotient space of ${\rm SO}(n, \mathbb{R})$ modulo the subgroup $\widetilde{{\rm SO}}(k, \mathbb{R})$ generated by ${\rm SO}(k, \mathbb{R})$ (orthogonal transformations of the subspace $\mathbb{R}^{k} \subset \mathbb{R}^{n}$ spanned by the f\/irst $k$-axes) and the subgroup of diagonal matrices in ${\rm SO}(n, \mathbb{R})$. This subgroup is equal to the intersection of ${\rm SO}(n, \mathbb{R})$ and the cartesian product
\begin{gather*}
{\rm O}(k)\times\underbrace{{\rm O}(1)\times\dots \times {\rm O}(1)}_{n-k\ \text{times}} ,
\end{gather*}
and the quotient space ${\rm SO}(n, \mathbb{R})/\widetilde{{\rm SO}}(k, \mathbb{R})$ is equal to the partial f\/lag space ${\rm Fl}_{k,1,\dots,1}$ (with $n-k$ units):
\begin{gather*}
{\rm Fl}_{k,1,\dots,1}=\big\{0\subset W\subset V_1\subset V_2\subset \dots\subset V_{n-k}=\mathbb R^n\big\},
\end{gather*}
where $\dim W=k$, $\dim V_i=k+i$.

Finally, consider the most general case. Assume that the eigenvalues of $\Lambda$ are divided into several ``clusters'':
$\lambda_1=\dots=\lambda_{i_1},\lambda_{i_1+1}=\dots=\lambda_{i_1+i_2},\dots,\lambda_{i_1+\dots+i_{k-1}}=\dots=\lambda_{i_1+\dots+i_k}$, where $i_1+i_2+\dots+i_k=n$. Then the orbit $\Psi\Lambda\Psi^{-1}$ will be equal to the partial f\/lag space ${\rm Fl}_{i_1,\dots,i_k}$:
\begin{gather*}
{\rm Fl}_{i_1,i_2,\dots,i_k}=\big\{0=V_0\subset V_1\subset V_2\subset V_3\subset \dots\subset V_k=\mathbb R^n\big\},
\end{gather*}
where $\dim V_l=i_1+i_2+\dots+i_l$. This can be proved either by the considerations of the symmetry (i.e., by f\/inding the subgroup of matrices commuting with $\Lambda$) as before or one can use the following observation: every symmetric matrix $L\in \operatorname{Symm}_n$ with eigenvalues $\lambda_1,\dots,\lambda_k$ of multiplicities $i_1,\dots,i_k$ is uniquely determined by the collection of eigenspaces:
\begin{gather*}
L_{j}=\big\{0\neq v\in\mathbb R^n\,|\, Lv=\lambda_jv\big\},\qquad j=1,\dots,k.
\end{gather*}
Clearly $\dim L_j=i_j$. So one naturally identif\/ies this matrix with a point in ${\rm Fl}_{i_1,\dots,i_k}$ by putting $V_l=L_1\oplus L_2\oplus\dots\oplus L_l$.

An important property of the partial f\/lag spaces is the existence of surjective projections
\begin{gather}\label{eq:projflag}
\pi\colon \ {\rm Fl}_n(\mathbb R)\to {\rm Fl}_{i_1,i_2,\dots,i_n}(\mathbb R)
\end{gather}
given by ``forgetting'' the unnecessary subspaces. Alternatively, these projections can be regarded as the additional factorization
\begin{gather*}
{\rm Fl}_n(\mathbb R)={\rm SO}(n, \mathbb R)/{\rm SO}(n, \mathbb R)\bigcap(\underbrace{{\rm O}(1)\times\dots\times {\rm O}(1)}_{n\ \text{times}}),\\
{\rm Fl}_{i_1,\dots,i_k}(\mathbb R)={\rm SO}(n, \mathbb R)/{\rm SO}(n, \mathbb R)\bigcap({\rm O}_{i_1}(\mathbb R)\times\dots\times {\rm O}_{i_k}(\mathbb R)).
\end{gather*}
It follows from this description that $\pi$ is a locally trivial f\/ibre bundle with the f\/ibre equal to
\begin{gather*}
X={\rm SO}(n, \mathbb R)\bigcap({\rm O}_{i_1}(\mathbb R)\times\dots\times {\rm O}_{i_k}(\mathbb R))/{\rm SO}(n, \mathbb R)\bigcap(\underbrace{{\rm O}(1)\times\dots\times {\rm O}(1)}_{n\ \text{times}}).
\end{gather*}
Or using the notation we introduced earlier
\begin{gather*}
X={\rm SO}(n, \mathbb R)\bigcap({\rm O}_{i_1}(\mathbb R)\times\dots\times {\rm O}_{i_k}(\mathbb R))/T_n^+.
\end{gather*}

\subsection[Bruhat order in $S_n$ and full f\/lag spaces]{Bruhat order in $\boldsymbol{S_n}$ and full f\/lag spaces}
\label{ss:bru:s_n}
In this paragraph we will closely follow the exposition in \cite{BB,F}, with only a few notations changed. Let $\omega$ be a permutation,
\begin{gather*}
\omega\in S_n,\qquad \omega\colon \ \{1,2,\dots,n\}\to\{1,2,\dots,n\}.
\end{gather*}
This permutation can be abbreviated to $(\omega(1),\dots,\omega(n))$. One def\/ines the \textit{length} of the permutation $\omega$ as the total number of involutions in the sequence $(\omega(1),\dots,\omega(n))$; that is,
\begin{gather*}%{Inv}
l_{\omega} = \# \{ j_{1} < j_{2} \,|\, \omega(j_{2}) < \omega(j_{1}) \}.
\end{gather*}
Let the numbers $r_\omega[p,q]$ be equal to the ``number of involutions with respect to $p$, $q$'':
\begin{gather}\label{Inv:pq}
r_{\omega}[p,q] = \# \{ j \leq p \,|\, \omega(j) \geq q \},\qquad 1 \leq p,q \leq n.
\end{gather}
There are many equivalent def\/initions of the Bruhat order on permutations; we give the following (see~\cite{BB}):
\begin{defi}\label{Bruhat}
The \textit{Bruhat order} on $S_n$ is the partial order determined by the following relation: for any two permutations $u$ and $v$ in $S_n$, one says that $u$ precedes $v$ ($u\prec v$) if and only if
\begin{gather*}
r_{u}[p,q] \leq r_{v}[p,q]\qquad \text{for all} \ p,\; q.
\end{gather*}
\end{defi}

A simple way to compute $r_{\omega}[p,q]$ is to consider the matrix $A_\omega$ representing the permuta\-tion~$\omega$:
\begin{gather*}
(A_\omega)_{ij}=\begin{cases}
 1,& \omega(j)=n-i+1,\\
 0,& \text{otherwise}.
 \end{cases}
\end{gather*}
It is clear that for all $p$, $q$ the number $r_{\omega}[p,q]$ is equal to the rank of the submatrix $A^{pq}_\omega=\left((A_\omega)_{ij}\right)_{i\le q,\,j\le p,}$ of $A_\omega$ (informally one can say that $A^{pq}_\omega$ is the submatrix in the upper left corner of $A_\omega$ determined by the element $a_{pq}$). Since there is only one nonzero entry in every row and column of $A_\omega$, it is enough to count the number of such nonzero elements in $A^{pq}_\omega$. E.g., for $\omega=(4213)$ we obtain $r_{\omega}[3,2]=2$ from the following table (we replace $1$ by $\bullet$):
\begin{gather*} %{T1}
\begin{tabular}{c|c|c|c|c|}
\hline\cline{1-0}
$4$ & $\bullet$ & $$ & $$ & $$\\
\hline\cline{1-0}
$3$ & $$ & $$ & $$ & $\bullet$\\
\hline\cline{1-0}
$2$ & $$ & $\bullet$ & $\leftarrow \uparrow$ & $$\\
\hline\cline{1-0}
$1$ & $$ & $$ & $\bullet$ & $$\\
\hline\cline{1-0}
$$ & $1$ & $2$ & $3$ & $4$\\
\end{tabular}
\end{gather*}
The following lemma is proved in~\cite{F}:
\begin{lem}\label{Bruhatlem}
Let $u \prec v$, $u \neq v$. Let $j$, $1\le j\le n$ be the smallest integer, for which $u(j) \neq v(j)$ $($and hence $u(j) < v(j))$. Let $n\ge k>j$ be the least integer for which $u(j) \leq v(k) < v(j)$ and let $v'=v\cdot(j,k)$ denote the composition of~$v$ with the swap of~$j$ and~$k$. Then $u \prec v' \prec v$.
\end{lem}

The relation between this order and the geometry of the full f\/lag space is determined by the structure of the Schubert cell decomposition of this space. There are many ways to def\/ine the Schubert cells, we shall use the following one (more def\/initions can be found in the literature, see \cite{BB, Br, F} and cf.\ Section~\ref{ss:bru:p&pfl}).

\begin{defi}\label{def:Schu1}
Embed the group $S_n$ into ${\rm SL}(n,\mathbb R)$ as it is explained above (one should only replace $1$ in the def\/inition of $A_\omega$ to $\pm 1$ so as to make sure that the determinant is equal to $1$ and not $-1$). Using the projection
\begin{gather*}
p\colon \ {\rm SL}(n,\mathbb R)\to {\rm Fl}_n(\mathbb R)= {\rm SL}(n,\mathbb R)/B_n^+,
\end{gather*}
we obtain the points $p(A_\omega)\in {\rm Fl}_n(\mathbb R)$, which we shall denote by $[A_\omega]$. The subgroup $B_n^+$ of~${\rm SL}(n,\mathbb R)$ acts on ${\rm Fl}_n(\mathbb R)$ and \textit{the Schubert cell $X_\omega\subset {\rm Fl}_n(\mathbb R)$ for $\omega\in S_n$} is the orbit of $[A_\omega]$ with respect to this action:
\begin{gather*}
X_\omega=B_n^+\cdot[A_\omega]\subseteq {\rm Fl}_n(\mathbb R).
\end{gather*}
\end{defi}

One can show (cf.~\cite{F} and Section~\ref{ss:bru:p&pfl}) that $X_\omega$ is indeed a cell. In ef\/fect,
\begin{gather*}
X_\omega\cong\mathbb R^{l(\omega)}.
\end{gather*}
The closure $\overline{X}_\omega$ of $X_\omega$ in the f\/lag space is a singular algebraic variety. It is called \textit{the Schubert variety}. The following statement explains the relation of the Schubert cells and the Bruhat order (see \cite[Proposition~7, p.~175]{F}):
\begin{predl}\label{prop:schubru}
An element $w\in S_n$ precedes $v\in S_n$ with respect to the Bruhat order, $w\prec v$, if and only if $\overline{X}_w\subseteq\overline{X}_v$.
\end{predl}

This proposition is sometimes used to give an alternative def\/inition of the Bruhat order: \textit{the Bruhat order is the partial order induced by the contiguity of the cells in the Schubert cell decomposition of the flag manifold}.

There is another important observation relating the Schubert cells and Bruhat order. Namely, consider the $B_n^-$-orbits of the same points $[A_\omega]$:
\begin{gather*}
X^\vee_\omega=B^-_n\cdot[A_\omega].
\end{gather*}
They are called \textit{the dual Bruhat cells}. These sets are as well homeomorphic to the Euclidean spaces, their closures are called \textit{the dual Schubert varieties}. Then the following is true (see again~\cite{F}):
\begin{predl}\label{prop:dualschu}
Let $v,w\in S_n$ be two elements. Then $X_v\bigcap X^\vee_w\neq\varnothing$ if and only if $v\prec w$ in Bruhat order. In the latter case the intersection of cells is transversal.
\end{predl}

This property of the Schubert cells was crucial in our description of the asymptotic behavior of the FS Toda lattice, see \cite{CSS}. Below we shall make use of analogous properties of the partial f\/lag spaces.

\subsection{Bruhat order on multiset permutations}\label{ss:bru:p&pfl}
In this section we give a description of the Schubert cell decomposition of the partial f\/lag space and the Bruhat order associated with it. Since we were not able to f\/ind a combinatoric description of this order in literature, we do not conf\/ine ourselves here to mere def\/initions and statement of results, and give proofs of a few facts (in particular, see Proposition~\ref{prop:Brurep1}). Besides this one can refer to~\cite{Br} and book~\cite{BB} for more details.

Let $n>k$ be two natural numbers and let $I=(i_1,i_2,\dots,i_k)$ be a partition of $n$ into $k$ parts, i.e., for each $j=1,\dots,k$ we let $i_j$ be a positive integer so that
\begin{gather*}
i_1+i_2+\dots+i_k=n.
\end{gather*}
Then we give the following def\/inition:
\begin{defi}\label{def:repet}
An \textit{$I$-permutation of multiset} (or permutation with multiplicities, or permutation with repetitions) is any surjective map
\begin{gather*}
\tau\colon \ \{1,2,\dots,n\}\to\{1,2,\dots,k\},
\end{gather*}
such that $|\tau^{-1}(j)|=i_j$ for all $j=1,\dots,k$. We shall denote the set of all such permutations by~$S_n^I$.
\end{defi}
One can regard an element $\tau\in S_n^I$ as a string of elements of the form
\begin{gather*}
\tau=(\tau(1),\tau(2),\dots,\tau(n)),
\end{gather*}
where $1\le \tau(j)\le k$ and every element $j$, $1\le j\le k$ appears exactly $i_j$ times (which explains the name of these objects that we use).

Below we shall make an extensive use of the following map $\tau_1^*\colon S_n\to S_n^I$, where $\tau_1\in S_n^I$ is the following multiset permutation which we regard as a map
\begin{gather*}
\tau_1\colon \ \{1,2,\dots,n\}\to\{1,2,\dots,k\},\\
\tau_1 (l)=\begin{cases}
 1,& 1\le l\le i_1,\\
 2,& i_1+1\le l\le i_1+i_2,\\
 &\dots\dots\dots\dots\dots\dots\dots\\
 j,& i_1+\dots+i_{j-1}+1\le l\le i_1+\dots+i_{j-1}+i_j,\\
 &\dots\dots\dots\dots\dots\dots\dots\dots\dots\dots\dots\dots\dots\dots\dots\\
 k,& i_1+\dots+i_{k-1}+1\le l\le n.
 \end{cases}
\end{gather*}
Then for every permutation $w\in S_n$ we def\/ine the multiset permutation $\tau^{\ast}_1(w)=\tau_w$ by the formula
\begin{gather*}
\tau_w=\tau_1\circ w ,
\end{gather*}
or in the ``linear form'':
\begin{gather*}
\tau_w=(\tau_1(w(1)),\tau_1(w(2)),\dots,\tau_1(w(n))).
\end{gather*}
The map $\tau_1^*$ is evidently an epimorphism and
\begin{gather*}
(\tau_1^*)^{-1}(\tau_1)=S_{i_1}\times S_{i_2}\times\dots\times S_{i_k}\subseteq S_n,
\end{gather*}
where we embed $S_{i_j}$ into $S_n$ by letting it permute the elements $l$, $i_1+\dots+i_{j-1}<l\le i_1+\dots+i_j$. For all other elements $\tau\in S_n^I$ we see that $(\tau_1^*)^{-1}(\tau)$ is a right coset of $S_n$ by the subgroup $S_{i_1}\times S_{i_2}\times\dots\times S_{i_k}$, i.e., the set of all permutations of the form $x_\tau\cdot w$ for any $w\in S_{i_1}\times S_{i_2}\times\dots\times S_{i_k}$ and some $x_\tau\in S_n$ such that $\tau_1^*(x_\tau)=\tau$. We shall sometimes call these cosets \textit{clusters in $S_n$, corresponding to $\tau\in S_n^I$}; clearly every such cluster contains $N=i_{1}!i_{2}!\cdots i_{k}!$ elements. One can describe elements in $S_n^I$ in terms of clusters; for instance, the string $(1,2,2,1)\in S_4^{(2,2)}$ corresponds to the cluster $\{(1,3,4,2),\,(2,3,4,1),\,(1,4,3,2),\,(2,4,3,1)\}$.

Now one uses these constructions to introduce the Bruhat order on $S_n^I$. Loosely speaking it is obtained by ``pulling back'' from $S_n$. However, this procedure is not always well def\/ined, so we give some details.

First of all, just as in Section~\ref{ss:bru:s_n} we introduce the notion of \textit{length} of a permutation with repeating elements: let $\tau\colon \{1,2,\dots,n\}\to \{1,2,\dots,k\}$ be an element in $S_n^I$, then
\begin{gather*}
l_\tau=\#\{j_1<j_2\,|\, \tau(j_2)>\tau(j_1)\},
\end{gather*}
i.e., $l_\tau$ is the number of involutions in the sequence $(\tau(1),\dots,\tau(n))$. Observe that we do not count the pairs $j_1<j_2$ for which $\tau(j_1)=\tau(j_2)$.

Second, similarly to the Def\/inition~\ref{Bruhat} we def\/ine the Bruhat order on $S_n^I$ with the help of ranks of some matrices: consider the $n\times k$ matrices $A_\tau$ corresponding to the elements of~$S_n^I$. Just as before we def\/ine the numbers $r_\tau[p,q]$ (see formula~\eqref{Inv:pq}) for all $1\le p\le n$ and $1\le q\le k$. Then:
\begin{defi}\label{Bruhat2}
We say that $\tau\prec\upsilon$ (where both $\tau,\upsilon\in S_n^I$) if and only if $r_\tau[p,q]\leq r_\upsilon[p,q]$ for all~$p$,~$q$.
Observe that we use $\leq$ rather than the strict inequality $<$ here. Of course at least one inequality must be strict if $\tau\neq\upsilon$
\end{defi}

The following example illustrates this def\/inition very well: let $I=(2,2)$ as before and $\tau=(2,1,1,2)$. Then $A_\tau$ has this shape:
\begin{gather*} %{T1}
\begin{tabular}{c|c|c|c|c|}
\hline\cline{1-0}
$2$ & $\bullet$ & $$ & $$ & $\bullet$\\
\hline\cline{1-0}
$1$ & $$ & $\bullet$ & $\bullet$ & $$\\
\hline\cline{1-0}
$$ & $1$ & $2$ & 3 & $4$\\
\end{tabular}
\end{gather*}
Similarly, the lengths of the elements of $S_4^{(2,2)}$ are equal to
\begin{gather*}%{weak-Inv-Bruhat-Gr24}
(1,1,2,2), \quad l_{(1,1,2,2)}=0,\qquad
(1,2,1,2), \quad l_{(1,2,1,2)}=1,\qquad
(1,2,2,1), \quad l_{(1,2,2,1)}=2,\\
(2,1,1,2), \quad l_{(2,1,1,2)}=2,\qquad
(2,1,2,1), \quad l_{(2,1,2,1)}=3,\qquad
(2,2,1,1), \quad l_{(2,2,1,1)}=4.
\end{gather*}
These data are suf\/f\/icient to draw the Hasse diagram of the Bruhat order in this case (cf.\ Fig.~\ref{bundle-2})\footnote{Recall that the Hasse diagram of a partially ordered set is the oriented graph whose vertices are given by the elements of the set and edges connect elements $a$ and $b$ if and only if $a \prec b$ and there is no $c$ such that $a \prec c \prec b$.}.

The following propositions (Propositions~\ref{prop:Brurep1} and~\ref{prop:Brurep2}, and Lemma~\ref{Bruhatlem2}) explain, in what sense one can ``pull back'' the order from~$S_n$ to~$S_n^I$.
\begin{predl}\label{prop:Brurep1}
Let $u,v\in S_n$ be two permutations such that $u \prec v$ in the Bruhat order on~$S_n$ and $\tau_1^*(u)\neq \tau_1^*(v)$. Then $\tau_1^*(u)\prec\tau_1^*(v)$ in the Bruhat order in~$S_n^I$.
\end{predl}

\begin{proof}
Let us denote $\tau_1^*(u)=U$ and $\tau_1^*(v)=V$; we will show that $U\prec V$. First of all we observe that it is enough to show this in a particular case
\begin{gather*}
v=u\cdot(i,j),\qquad l(v)=l(u)+1
\end{gather*}
(here $(i,j)$ denotes the transposition of $i$ and $j$ for $i<j$). Indeed, it follows directly from Lemma~\ref{Bruhatlem} that every pair of Bruhat-comparable elements in $S_n$ can be connected by a ``path'' of intermediate elements verifying this condition at every stage.

So let $u$ and $v$ be as above. Since $u\prec v$ we have $r_{u}[p,q] \leq r_{v}[p,q]$. We shall show that $r_{U}[p,q] \leq r_{V}[p,q]$ for all $1\le p\le n$, $1\le q\le k$. It is impossible that all the numbers $r_U[p,q]$ and $r_V[p,q]$ are equal (we assumed that $U\neq V$) and so we have that $U\prec V$. Recall now that we assumed that $v=u(i,j)$, $i<j$ and $l(v)=l(u)+1$. This means, in particular, that $u(i) < u(j)$ (otherwise the number of inversions would decrease from multiplication by $(i,j)$). Let $u(i)=a$, $u(j)=b$ then $v(i)=b$, $v(j)=a$; let $i<l<j$ and for $c=u(l)$ there are three possibilities: $c < a < b$, $a < b < c$ or $a < c < b$. Counting the number of inversions in all three cases, we see that only the f\/irst two options are possible (in the third case the number of inversions will change at least by~$3$ when we swap~$i$ and~$j$).

In terms of the rectangular matrices associated with the elements of $S_n$ we can say that the rectangle $P_{abij}$ cut from $A_u$ by the $a$-th and $b$-th rows and $i$-th and $j$-th columns contains no nonzero elements, see~\eqref{Tabuv}.
\begin{gather}\label{Tabuv}
\begin{tabular}{c|c|c|c|c|c|c|c|c|c|c|c|}
\hline\cline{1-0}
$n$ & $$ & $$ & $$ & $$ & $$ & $$ & $$ & $$ & $$ & $$ & $$\\
\hline\cline{1-0}
$\dots $ & $$ & $$ & $$ & $$ & $$ & $$ & $$ & $$ & $$ & $$ & $$\\
\hline\cline{1-0}
$b$ & $$ & $$ & $$ & $$ & $$ & $$ & $$ & $$ & $1$ & $$ & $$\\
\hline\cline{1-0}
$$ & $$ & $$ & $$ & $$ & $$ & $$ & $$ & $$ & $$ & $$ & $$\\
\hline\cline{1-0}
$$ & $$ & $$ & $$ & $$ & $$ & $$ & $$ & $$ & $$ & $$ & $$\\
\hline\cline{1-0}
$$ & $$ & $$ & $$ & $$ & $$ & $$ & $$ & $$ & $$ & $$ & $$\\
\hline\cline{1-0}
$a$ & $$ & $$ & $$ & $1$ & $$ & $$ & $$ & $$ & $$ & $$ & $$\\
\hline\cline{1-0}
$\dots $ & $$ & $$ & $$ & $$ & $$ & $$ & $$ & $$ & $$ & $$ & $$\\
\hline\cline{1-0}
$2$ & $$ & $$ & $$ & $$ & $$ & $$ & $$ & $$ & $$ & $$ & $$\\
\hline\cline{1-0}
$1$ & $$ & $$ & $$ & $$ & $$ & $$ & $$ & $$ & $$ & $$ & $$\\
\hline\cline{1-0}
$$ & $1$ & $2$ & $\dots $ & $i$ & $$ & $$ & $$ & $$ & $j$ & $\dots $ & $n$
\end{tabular}
\end{gather}
Clearly, no nonzero elements will appear in this rectangle when we swap~$i$ and~$j$.

When we pass from $u$ and $v$ to $U$ and $V$, this matrix ``shrinks'' vertically and the interior of the corresponding rectangle $P_{\tau_1(a)\tau_1(b)ij}$ in the matrices $A_U$ and $A_V$ does not contain nonzero elements: new nonzero elements can appear only in the top and bottom rows of the rectangle (i.e., in the rows number $\tau_1(a)$ and $\tau_1(b)$) and this can happen only simultaneously for~$U$ and~$V$. Hence the matrices $A_U$ and $A_V$ dif\/fer only by positions of~$1$'s in the corners of $P_{k_{a}k_{b}ij}$: in $A_U$ they stand at the entries $(k_{a},i)$ and $(k_{b},j)$, while in $A_V$ they stand at $(k_{a},j)$ and $(k_{b},i)$. Now a~simple inspection shows that $r_{U}[p,q]<r_{V}[p,q]$ for all $p$ and $q$, see~\eqref{TabUV2}.
\begin{gather} \label{TabUV2}
A_U=
\begin{tabular}{c|c|c|c|c|c|c|c|}
\hline\cline{1-0}
$\dots $ & $\dots $ & $\dots $ & $\dots $ & $\dots $ & $\dots $ & $\dots $ & $\dots $\\
\hline\cline{1-0}
$k_{b}$ & $$ & $$ & $$ & $\dots $ & $$ & $1$ & $$\\
\hline\cline{1-0}
$$ & $$ & $$ & $$ & $\dots $ & $$ & $$ & $$\\
\hline\cline{1-0}
$\dots $ & $\dots $ & $\dots $ & $\dots $ & $\dots $ & $\dots $ & $\dots $ & $\dots $\\
\hline\cline{1-0}
$$ & $$ & $$ & $$ & $\dots $ & $$ & $$ & $$\\
\hline\cline{1-0}
$k_{a}$ & $$ & $1$ & $$ & $\dots $ & $$ & $$ & $$\\
\hline\cline{1-0}
$\dots $ & $\dots $ & $\dots $ & $\dots $ & $\dots $ & $\dots $ & $\dots $ & $\dots $\\
\hline\cline{1-0}
$$ & $\dots $ & $i$ & $$ & $\dots $ & $$ & $j$ & $\dots $
\end{tabular}
\\
A_V=
\begin{tabular}{c|c|c|c|c|c|c|c|}
\hline\cline{1-0}
$\dots $ & $\dots $ & $\dots $ & $\dots $ & $\dots $ & $\dots $ & $\dots $ & $\dots $\\
\hline\cline{1-0}
$k_{b}$ & $$ & $1$ & $$ & $\dots $ & $$ & $$ & $$\\
\hline\cline{1-0}
$$ & $$ & $$ & $$ & $\dots $ & $$ & $$ & $$\\
\hline\cline{1-0}
$\dots $ & $\dots $ & $\dots $ & $\dots $ & $\dots $ & $\dots $ & $\dots $ & $\dots $\\
\hline\cline{1-0}
$$ & $$ & $$ & $$ & $\dots $ & $$ & $$ & $$\\
\hline\cline{1-0}
$k_{a}$ & $$ & $$ & $$ & $\dots $ & $$ & $1$ & $$\\
\hline\cline{1-0}
$\dots $ & $\dots $ & $\dots $ & $\dots $ & $\dots $ & $\dots $ & $\dots $ & $\dots $\\
\hline\cline{1-0}
$$ & $\dots $ & $i$ & $$ & $\dots $ & $$ & $j$ & $\dots $
\end{tabular}\tag*{\qed}
\end{gather}
\renewcommand{\qed}{}
\end{proof}

The following statement is similar to Lemma~\ref{Bruhatlem}; it gives a procedure by which we can move between comparable elements of~$S_n^I$.

\begin{lem}\label{Bruhatlem2}
Let $U,V\in S_n^I$, $U \prec V$, $U \neq V$. Let $j_{0}$ be the least $j$, $1\le j\le n$, for which $U(j_{0}) \neq V(j_{0})$; then $U(j_{0}) < V(j_{0})$. Let $l\ge j_0$ be the maximal number for which $V(j_{0})=V(j_{0}+1)=\dots =V(l)$ so that $U(j_{0}) < V(l)$. Let $m$ denote the least integer $l\le m\le n$ verifying the inequality $U(j) \leq V(m) < V(l)$. Put $V' = V(l,m)$, i.e., in terms of the corresponding strings $V'$ is $V$ in which the values of $V(l)$ and $V(m)$ are swapped. Then $U \prec V' \prec V$.
\end{lem}

The proof is by direct inspection of the def\/initions and we omit it. Finally we use this procedure to move ``in the opposite direction'': from $S_n^I$ to~$S_n$.
\begin{predl}\label{prop:Brurep2}
Every coset $(\tau_1^*)^{-1}(\upsilon)\subset S_n$, $\upsilon\in S_n^I$ contains the unique least element $u\in S_n$ (with respect to the Bruhat order in $S_n$) so that if $\xi\prec \zeta$ in $S_n^I$, then a similar inequality holds for the corresponding least elements in~$S_n$: $x\prec z$.
\end{predl}

Recall that the \textit{least element} in a partial order is the element, which is less than all others, in particular, it is comparable with all others. It is clear that it is unique if it exists.

\begin{proof} Observe that every coset $(\tau_1^*)^{-1}(\upsilon)$ contains a~unique element $s_\upsilon$ such that~$s_\upsilon^{-1}$ is a~\textit{$I$-shuffle}, i.e., a~permutation which preserves the order of the elements inside the ``blocks'' of the partition~$I$. We shall call such $s_\upsilon$ the \textit{$I$-deshuffle}; clearly, this deshuf\/f\/le is the least element that we need. This follows easily from the fact that for any permutation $u\in(\tau_1^*)^{-1}(\upsilon)$ and any $i<k<j$ such that $\upsilon(i)=\upsilon(j)\neq\upsilon(k)$ we shall have either $u(k)<u(i)$, $u(k)<u(j)$, or $u(i)<u(k)$, $u(j)<u(k)$ (i.e., $u(k)$ cannot lie between~$u(i)$ and~$u(j)$). Hence, the order in $(\tau_1^*)^{-1}(\upsilon)$ (induced from $S_n$) depends only on the permutation of the elements inside separate blocks of the partition, and not on~$\upsilon$. In particular, the shuf\/f\/le corresponds to the ``unit'' with respect to the permutations inside the blocks; hence, it is least.

Consider now arbitrary $U,V\in S_n^I$ and let $s_U$, $s_V$ be the corresponding deshuf\/f\/les. By Lemma~\ref{Bruhatlem2} we can assume that $V$ is obtained from $U$ by swapping two elements chosen as it is explained in its condition: $V=U(l,m)$. Then, as one can see $s_V=s_U(l,m)$: clearly $s_U(l,m)\in(\tau_1^*)^{-1}(V)$ and the fact that it coincides with~$s_V$ follows from the choice of~$l$ and~$m$, see Lemma~\ref{Bruhatlem2}. Similarly, from the same choice it follows that all the elements in the string corresponding to~$s_U$ that lie between $l$ and $m$, are either greater than both~$s_U(l)$ and~$s_U(m)$, or less than both. Hence,
\begin{gather*}
l_{s_V}=l_{s_U(l,m)}=l_{s_U}+1.\tag*{\qed}
\end{gather*}
\renewcommand{\qed}{}
\end{proof}

Observe that we have in fact established that the Bruhat order on $S_n$, when restricted to \textit{clusters} $(\tau_1^*)^{-1}(\upsilon)$, $\upsilon\in S_n^I$, coincides with the \textit{lexicographic order} on the direct product $S_{i_1}\times S_{i_2}\times\dots\times S_{i_k}$ where each factor is equipped with the usual Bruhat order on~$S_l$. In particular, there is a unique \textit{greatest} element in it too.

\subsection{Schubert cells in partial f\/lags and Bruhat order}\label{ss:bru:schucells}
The def\/inition of the Schubert cells given in Section~\ref{ss:bru:s_n} can be extended to a more general situation of the partial f\/lag spaces, in particular to Grassmannians, so that the canonical projection~\eqref{eq:projflag} maps the Schubert cells in ${\rm Fl}_n(\mathbb R)$ to the cells in the partial f\/lags. We shall use this observation and the properties of the Bruhat order to describe some of the structures of the Schubert cells, their duals and their intersections in the partial f\/lag spaces. The main references for this section are~\cite{Br,F, GH}.

We begin with the general def\/inition of the Schubert cells in an important particular case: $k=1$, i.e., the f\/lag space is equal to a Grassmannian ${\rm Gr}_{d,n}(\mathbb R)$, the space of $d$-planes in $\mathbb R^n$.

Let $\{e_{1},\dots ,e_{n}\}$ be a basis of $\mathbb R^n$. Any $d$-dimensional hyperplane inside $\mathbb R^n$ is determined by a collection of $d$ linearly-independent vectors in $\mathbb R^n$, which is the same (when the basis of~$\mathbb R^n$ is f\/ixed) as a $d \times n$, $d<n$ matrix with a maximal rank. It is clear that two matrices~$\Lambda$,~$\Lambda'$ determine the same $d$-plane if and only if there exists $g \in {\rm GL}(d, \mathbb{R})$ such that:
\begin{gather*}%{defplane}
\Lambda' = g \cdot \Lambda.
\end{gather*}
We shall use the same symbol $\Lambda$ to denote the $d$-plane corresponding to the matrix~$\Lambda$. Observe that the natural ${\rm GL}(n, \mathbb R)$-action on the Grassmannian translates into the right multiplication of $\Lambda$ by the matrices $B\in {\rm GL}(n, \mathbb R)$. Also observe that in our notation the left and right actions exchange their roles: in Section~\ref{ss:bru:p&pfl} we def\/ine ${\rm Gr}_{d,n}$ as the quotient of ${\rm SL}(n, \mathbb R)$ by the right action of a parabolic subgroup so that ${\rm GL}(n, \mathbb R)$ acts on it from the left, and when we draw rectangular matrices, the action and factorization change directions. This can be amended by the use of the $n\times d$ matrices instead of the $d\times n$ matrices.

Let for any multi-index $A=(a_1,\dots,a_d)$, $1\le a_1<a_2<\dots<a_d\le n$ symbol $V_A$ denote the subspace in $\mathbb R^n$ spanned by $\{e_{a_1},e_{a_2},\dots,e_{a_d}\}$. Then
\begin{defi}\label{def:Schugra}
The Schubert cell $X_A$ is the $B_n^+$-orbit of the point $V_A$ in Grassman\-nian ${\rm Gr}_{d,n}(\mathbb R)$.
\end{defi}

One can give a more geometric def\/inition of the Schubert cells. To this end consider the subspaces $V_i$, $i=1,\dots,n$, $V_i\subseteq\mathbb R^n$ spanned for each $i$ by the vectors $\{e_{1},\dots ,e_{i}\}$ and consider for an arbitrary $d$-plane $\Lambda\in {\rm Gr}_{d,n}(\mathbb R)$ the intersections:
\begin{gather*}%{seqintersec}
0 \subset \Lambda \cap V_{1} \subset \Lambda \cap V_{2} \subset \cdots \subset \Lambda \cap V_{n-1} \subset \Lambda \cap V_{n} = \Lambda.
\end{gather*}
Then for the same multi-index $A=(a_1,\dots,a_d)$ we have
\begin{defi}\label{def:Schugra2}
The Schubert cell $X_A$ is the set of such planes $\Lambda\in {\rm Gr}_{d,n}(\mathbb R)$ that
\begin{gather*}
\dim (\Lambda\cap V_j)=\#\{k\,|\, 1\le k\le d,\ a_k<j\},\qquad \text{for all}\ j=1,\dots,n.
\end{gather*}
\end{defi}

The rectangular matrices $\Lambda$ corresponding to the planes in $X_A$ can be chosen in a special form:
\begin{gather*}
\left(\begin{array}{@{}cccccccccccc@{}}
 0 & \cdots & 0 & 1 & {*} & \cdots & 0 & {*} & 0 & \cdots & \cdots & {*} \\
 0 & \cdots & 0 & 0 & \cdots & 0 & 1 & {*} & 0 & \cdots & \cdots & {*} \\
\cdots & \cdots & \cdots & 0 & \cdots & \cdots & 0 & \cdots & 0 & \cdots & \cdots & {*} \\
 0 & \cdots & \cdots & 0 & \cdots & \cdots & 0 & 0 & 1 & {*} & \cdots & {*}
\end{array}\right).
\end{gather*}
Here $1$'s stand in the intersections of the $i$-th rows and $(n-a_{d-i+1}+1)$-th columns and are the only nonzero elements of the column; asterisks are used to denote arbitrary real numbers. This f\/lip of indices is caused by the fact that we have substituted the left action of $B_n^+$ for the right.

As one can see the sets $X_A$ are indeed cells: it follows from the shape of the matrices we use here that
\begin{gather*}%{hom}
X_A \cong R^{dn - \sum(n+i- a_{i})}.
\end{gather*}
Here $dn$ is the dimension of the space of the $d\times n$ matrices and we subtract the number of f\/ixed elements in a matrix. One can also introduce the ``length'' of the sequence~$A$ by putting
\begin{gather*}
l(A)=\sum_{k=1}^da_k-k.
\end{gather*}
Then it is easy to see that $l(A)=dn-\sum(n+i-a_i)$ and so
$X_A\cong \mathbb R^{l(A)}$.
It turns out that $\bigcup_AX_A$ is a cell decomposition of ${\rm Gr}_{d,n}(\mathbb R)$. Moreover, one can show (see \cite{Br}) that a cell $X_B$ is adjacent to $X_A$ if and only if $b_k\le a_k$ for all $k$. In this case we shall write $B\le A$; this is a~partial order on the set of all sequences.

Finally, we associate a sequence $A$ with the multiset permutation $\omega\in S_n^{(d,n-d)}$ given by
\begin{gather*}
\omega(a_1)=\dots=\omega (a_d)=1,\qquad \omega(j)=2,\qquad j\not\in\{a_1,\dots,a_d\}.
\end{gather*}
It is easy to see that the Bruhat order on $S_n^{(d,n-d)}$ corresponds to the partial order given by the inequalities $B\le A$.

\begin{ex}
Consider the Grassmannian ${\rm Gr}_{2,4}(\mathbb{R})$, the space of all real $2$-planes in $\mathbb{R}^4$. In this case there are $6$ cells corresponding to the sequences $(1,2)$, $(1,3)$, $(1,4)$, $(2,3)$, $(2,4)$ and~$(3,4)$. The corresponding Schubert cells are spanned by the matrices of the following shapes (we use the equality signs to denote this):
\begin{gather*}%{mmll-Sch}
C_{(1,2)}=
\left(
\begin{matrix}
 0 & 0 & 1 & 0 \\
 0 & 0 & 0 & 1
\end{matrix}
\right), \qquad
C_{(1,3)}=
\left(
\begin{matrix}
 0 & 1 & * & 0 \\
 0 & 0 & 0 & 1
\end{matrix}
\right), \qquad
C_{(1,4)}=
\left(
\begin{matrix}
 1 & * & * & 0 \\
 0 & 0 & 0 & 1
\end{matrix}
\right), \\
C_{(2,3)}=
\left(
\begin{matrix}
 0 & 1 & 0 & * \\
 0 & 0 & 1 & *
\end{matrix}
\right), \qquad
C_{(2,4)}=
\left(
\begin{matrix}
 1 & * & 0 & * \\
 0 & 0 & 1 & *
\end{matrix}
\right), \qquad
C_{(3,4)}=
\left(
\begin{matrix}
 1& 0 & * & * \\
 0 & 1 & * & *
\end{matrix}
\right).
\end{gather*}
\end{ex}

In the general case of an arbitrary partial f\/lag space ${\rm Fl}_{i_1,\dots,i_k}(\mathbb R)$ (which we shall often abbreviate to just ${\rm Fl}_I(\mathbb R)$) we can def\/ine the Schubert cells decomposition similarly by choosing a~basis $\{e_{1},\dots ,e_{n}\}$ of~$\mathbb R^n$; then for any partial f\/lag
\begin{gather*}
E_\bullet=(E_1\subset E_2\subset\dots\subset E_k)
\end{gather*}
of vector subspaces in $\mathbb R^n$ and any multiset permutation $\omega\in S_n^I$ (where $I$ is a $k$-partition of $n$) we say that $E_\bullet$ is in the cell $X_\omega$ if and only if
\begin{gather}\label{cellflag}
\dim (E_{p} \cap E_{q}) = \# \{i \leq p \,|\, \omega(i) \leq q \}
\end{gather}
for all $1\le p\le n$, $1\le q\le k$. It is easy to see that these sets are indeed cells, that the cells def\/ined for dif\/ferent $\omega\in S_n^I$ do not intersect and that ${\rm Fl}_I(\mathbb R)$ is equal to their union. The closures of these cells will be denoted by~$\overline X_\omega$; they are called \textit{the Schubert varieties} in ${\rm Fl}_I(\mathbb R)$. One can describe them by replacing the equalities in~\eqref{cellflag} by the non-strict inequality~$\ge$; they are indeed singular algebraic subvarieties in~${\rm Fl}_I(\mathbb R)$.

As before, the same sets can be described as orbits of the group of upper-triangular matri\-ces~$B_n^+$; if we view the partial f\/lag space as the quotient space of ${\rm SL}(n, \mathbb R)$ (or ${\rm SO}(n, \mathbb R)$) by some parabolic subgroup~$P$, we can def\/ine the Schubert cells as the orbits of certain elements within~${\rm Fl}_I(\mathbb R)$ with respect to the action of $B^+_n$ (or some block-diagonal matrices, if we speak about ${\rm SO}(n, \mathbb R)$): take the matrix~$A_w$ corresponding to the minimal representative~$w$ in~$S_n$ of an element $\omega\in S_n^I$ (see Section~\ref{ss:bru:p&pfl}), then
\begin{gather*}
X_\omega=B^+_n(A_w)P/P\subset {\rm SL}(n;\mathbb R)/P.
\end{gather*}
It clearly follows from this def\/inition that \textit{the Schubert cells in the partial flag spaces are equal to the projections of the Schubert cells in the full flag space under the projection $\pi$, see~\eqref{eq:projflag}}. Moreover, one can show (see~\cite{Br}) that they are in ef\/fect homeomorphic to the cell of the least element~$w$ in ${\rm Fl}_n(\mathbb R)$.

Now by comparing this description with the results of the previous section (see Propositions~\ref{prop:Brurep1} and~\ref{prop:Brurep2} and Lemma~\ref{Bruhatlem2}) we obtain the following statement:
\begin{predl}\label{prop:Bruordcel}
For any two permutations with repetitions $\psi, \omega\in S_n^I$ we have $\psi\prec\omega$ in the Bruhat order if and only if the corresponding cells~$X_\psi$ and $X_\omega$ are adjacent; on the level of Schubert varieties this implies the inclusion $\overline{X}_\psi\subset\overline{X}_\omega$.
\end{predl}

Finally, similarly to the case of the full f\/lag space ${\rm Fl}_n(\mathbb R)$ in ${\rm Fl}_I(\mathbb R)$ one can also def\/ine the dual Schubert cells $X^\wedge_\omega$, $\omega\in S_n^I$ (and the corresponding Schubert varieties). This can be done either by a modif\/ication of the rank (in)equalities~\eqref{cellflag}, or as orbits of the corresponding elements in ${\rm Fl}_I(\mathbb R)$ with respect to the action of $B^-_n$, or f\/inally simply as projections of the dual cells in~${\rm Fl}_n(\mathbb R)$. As before a quick comparison of the def\/initions proves the following proposition:
\begin{predl}\label{prop:dualcells}
A Schubert cell $X_\psi\subset {\rm Fl}_I(\mathbb R)$ intersects with the dual Schubert cell $X^\wedge_\omega\subset {\rm Fl}_I(\mathbb R)$ if and only if $\psi\prec\omega$ in the Bruhat order in $S_n^I$ and in the latter case the intersection is transversal.
\end{predl}

The proof follows from the observation that $X^\wedge_\omega$ corresponds to the cell determined by the greatest element $w'$ in the cosets $(\tau_1^*)^{-1}(\omega)$ just as $X_\omega$ corresponds to the least one: i.e., $X_\omega^\wedge$ is a homeomorphic image of the dual cell corresponding to $w'$ in ${\rm Fl}_n(\mathbb R)$ under the projection $\pi$. Since the statement of the Proposition~\ref{prop:dualcells} holds in ${\rm Fl}_n(\mathbb R)$, it must hold in ${\rm Fl}_I(\mathbb R)$ as well (we must also use the fact that $\pi$ is submersion).

\section{The FS Toda lattice on partial f\/lag spaces}\label{s:tod}
\subsection{Fibre bundles and Bott--Morse functions}\label{ss:tod:mbott}
In our study of the Toda system on symmetric matrices with coinciding eigenvalues we shall need a few statements from the dif\/ferential geometry of homogeneous spaces of Lie groups. We prove them here for the sake of completeness.

Let $M$ be a smooth closed manifold with a smooth right transitive action of a compact Lie group $\tilde G$, $h\in\tilde G$, $x\in M$, $x\mapsto x^h$. Let $G\subset \tilde G$ be a compact subgroup of $\tilde G$. We assume that $\pi\colon M\to M/G$ is a locally trivial bundle (in which the f\/ibre containing $x$ is homeomorphic to the stabilizer of~$x$, $G^x\subset G$); in particular, we assume that $X=M/G$ is a smooth manifold. Let~$g$ be a $G$-invariant Riemannian metric on~$M$ and~$f$ be a $G$-invariant smooth function. In this case the gradient vector f\/ield $\xi=\operatorname{grad}_g f$ is $G$-invariant as well. Since~$g$ and $f$ are $G$-invariant they induce a Riemannian structure $g_X$ and a smooth function $f_X$ on $X$; similarly, the vector f\/ield $\xi$ induces a f\/ield $\xi_X$ on $X$; we just put $\xi_X(\pi(m))$ to be equal to the projection $d\pi_m(\xi(m))$ for an arbitrary point~$m$ in~$M$.

We shall need the following statement:
\begin{predl}\label{refgrad}
The field $\xi_X$ is equal to the gradient of the function $f_X$ with respect to the metric $g_X$:
\begin{gather*}
\xi_X=\operatorname{grad}_{g_X}f_X.
\end{gather*}
\end{predl}
\begin{proof}
The statement follows almost immediately from the def\/initions of a gradient vector f\/ield and induced metric $g_X$ on $M/G$: recall that the f\/ield $\operatorname{grad}_{g_X}f_X$ is characterized by the equality
\begin{gather*}
g_X(\eta_X, \operatorname{grad}_{g_X}f_X)=\eta_X(f_X);
\end{gather*}
and that $g_X(\eta_X,\,\zeta_X)$ for two vectors $\eta_X, \zeta_X\in T_xX$ (here $X=M/G$) is def\/ined as
\begin{gather*}
g_X(\eta_X, \zeta_X)=g(\eta,\zeta),
\end{gather*}
for any $\eta, \zeta\in T_mM$ ($m\in\pi^{-1}(x)$ is an arbitrary point), such that $d\pi_m(\eta)=\eta_X$, $d\pi_m(\zeta)=\zeta_X$ and they are both perpendicular to the ``vertical'' subspace $T_m(mG)$ of $T_mM$.

Now since the function $f$ is constant along the f\/ibres of $\pi$, the f\/ield $\operatorname{grad}_gf$ is orthogonal to the vertical subspaces; so
\begin{gather*}
g_X(\eta_X, d\pi(\operatorname{grad}_gf))=g(\eta, \operatorname{grad}_gf)=\eta(f)=\eta_X(f_X),
\end{gather*}
where the second equality follows from the def\/inition of $\operatorname{grad}_gf$, and the third one from the def\/inition of~$f_X$ (it is enough to consider the local structure of a cartesian product in~$M$).
\end{proof}

Below we shall need to know what conditions should be imposed on $f$ to provide that~$f_X$ is a Morse function. Observe that in the general case $f$ cannot be a Morse function itself unless the group~$G$ is discrete (and f\/inite): because of the $G$-invariance of $f$ every critical point of $f_X$ in $X=M/G$ will correspond to a $G$-orbit consisting of critical points of $f$ in $M$. So let us give a simple criterion for $f_X$ to be Morse:
\begin{predl}\label{morseref}
Let $f$ be a $G$-invariant function and $g$ be a $G$-invariant metric on~$M$. The function $f_X\in C^\infty(M/G)$, induced from~$f$, is a Morse function if and only if at every critical point~$m$ of~$f$ the Hessian~$H(f)$ is nondegenerate when restricted to the orthogonal complement to the space of ``vertical'' vectors $($i.e., vectors tangent to the orbit$)$.
\end{predl}

\begin{proof}
First, recall that the gradient of a function is always perpendicular to its level sets. So in the case we consider it to be perpendicular to the $G$-orbits. It follows that it is not in the kernel of the projection $d\pi_m\colon T_mM\to T_{\pi(m)}M/G$ (as we observed earlier, the gradient of $f_X$ is equal to the projection of the gradient of~$f$). Hence the critical points of $f_X$ correspond to the ``critical orbits'' of~$f$ in~$M$, i.e., to the $G$-orbits consisting of the critical points of~$f$.

Second, Hessian of $f_X$ at a critical point $x_0\in X$ is the symmetric quadratic form on tangent space given by the formula
\begin{gather*}
H(f_X)_{x_0}(\xi(x_0), \eta(x_0))=\xi(\eta(f_X))(x_0),
\end{gather*}
where $\xi$, $\eta$ are two vector f\/ields in the neighborhood of~$x_0$: it is easy to show that the right-hand side of this formula is indeed symmetric $2$-form on vectors in~$T_{x_0}X$ if~$x_0$ is singular (in particular, it depends only on the values of~$\xi$ and $\eta$ at the point~$x_0$).

Now the proposition follows from the fact that every vector f\/ield~$\xi$ on $M/G$ can be locally lifted to a $G$-invariant vector f\/ield $\tilde\xi$ in the neighborhood of $\pi^{-1}(x_0)\subset M$ orthogonal to the f\/ibre $\pi^{-1}(x_0)$: to see this it is enough to use the local trivialization $\pi^{-1}(U)=U\times G^{x_0}$ for some open neighborhood~$U$ of~$x_0$. So we have the following formula:
\begin{gather*}
H(f_X)_{x_0}(\xi(x_0), \eta(x_0))=\tilde\xi(\tilde\eta(f))(\tilde{x}_0)=H(f)_{\tilde x_0}(\tilde\xi(\tilde x_0), \tilde\eta(\tilde x_0)),
\end{gather*}
where $\tilde x_0$ is a point in the f\/ibre $\pi^{-1}(x_0)$; the right-hand side of this formula is $G$-equivariant, so the statement of the proposition follows.
\end{proof}

As a matter of fact the condition of Proposition~\ref{morseref} is a variation of the Morse--Bott condition. Recall that a function $f\in C^\infty(X)$ ($X$ is a smooth manifold) is called \textit{the Morse--Bott function} if its critical set is a smooth (not connected) submanifold in $X$ and the Hessian is nondegenerate in the normal direction to this submanifold. In our case the set of critical points is equal to the $G^x$-orbits through critical points $x$ of $f$; hence, it is smooth (by assumption on the nature of the action) and normal direction is the direction orthogonal to the vertical vectors.

\subsection{The FS Toda lattice on partial f\/lags}\label{ss:tod:pfl}
As we have explained in the introduction, the FS Toda lattice in dimension $n$ induces a gradient f\/low on the orthogonal group ${\rm SO}(n, \mathbb R)$, which we here shall call \textit{the Toda flow} or just \textit{Toda system} on ${\rm SO}(n, \mathbb R)$. This f\/low is determined by the vector f\/ield
\begin{gather}\label{eq:todafield}
M(\Psi)=\big(\big(\Psi\Lambda\Psi^{-1}\big)_+-\big(\Psi\Lambda\Psi^{-1}\big)_-\big)\Psi,
\end{gather}
where $\Lambda$ is the diagonal matrix of eigenvalues of the symmetric Lax matrix $L$; in addition we identify the tangent space $T_\Psi {\rm SO}(n,\mathbb R)$ with the right translation of the Lie algebra $\mathfrak{so}_n(\mathbb R)$ by~$\Psi$. In fact, this vector f\/ield is equal to the gradient of a function with respect to an invariant metric on ${\rm SO}(n, \mathbb R)$: for the f\/ixed eigenvalues matrix $\Lambda$ one can take
\begin{gather}\label{eq:gradfun}
F(\Psi)=\operatorname{Tr}\big(\Psi\Lambda\Psi^TN\big),\qquad \text{where}\quad N=\operatorname{diag}(0,1, \dots , n -1).
\end{gather}
The ${\rm SO}(n, \mathbb R)$-invariant Riemannian structure is determined by its values on $\mathfrak{so}_n$, where it is given by the formula
\begin{gather}\label{eq:riemstru}
\langle A, B\rangle_J =-\operatorname{Tr}\big(AJ^{-1}(B)\big),
\end{gather}
for any antisymmetric matrices $A$ and $B$ and a linear isomorphism $J\colon\mathfrak{so}_n\to\mathfrak{so}_n$. Then
\begin{gather*}
M(\Psi)=\operatorname{grad}_{\langle\,,\,\rangle_J}\,F,
\end{gather*}
see \cite{CSS,dMP} for details. This property does not depend on the eigenvalues of~$L$, in particular it holds for both distinct and coinciding eigenvalues.

Since all the objects here are $T^+_n$-invariant (see Section~\ref{ss:bru:fs&gr}) the same formulas \eqref{eq:todafield}, \eqref{eq:gradfun} and \eqref{eq:riemstru} induce a vector f\/ield, a function and a Riemannian structure on the full f\/lag space
\begin{gather*}
{\rm Fl}_n(\mathbb R)={\rm SO}(n;\mathbb R)/T_n^+.
\end{gather*}
It follows from the discussion of Section~\ref{ss:tod:mbott} that the f\/ield $M$ on ${\rm Fl}_n(\mathbb R)$ is equal to the gradient of the corresponding function. It can be shown (see~\cite{CSS,dMP}) that the function $F$ is a~Morse function on both ${\rm SO}(n,\mathbb R)$ and the f\/lag space ${\rm Fl}_n(\mathbb R)$.

In a similar way one can use the propositions of Section~\ref{ss:tod:mbott} to construct gradient f\/lows on the partial f\/lag spaces. Suppose there are coinciding eigenvalues of $\Lambda$; for def\/initeness we can assume that they are
\begin{gather}
\lambda_1 =\lambda_2=\dots=\lambda_{i_1}<\lambda_{i_1+1}=\lambda_{i_1+2}=\dots=\lambda_{i_1+i_2}<\cdots\nonumber\\
\hphantom{\lambda_1}{} <\lambda_{i_1+\dots+i_{k-1}+1}=\lambda_{i_1+\dots+i_{k-1}+2}=\dots=\lambda_{i_1+\dots+i_{k-1}+i_k},\label{eq:eige}
\end{gather}
where $i_1+i_2+\dots+i_k=n$; that is we assume there are $k<n$ distinct eigenvalues of $\Lambda$ and the multiplicity of the $j$-th eigenvalue is~$i_j$. In this case the vector f\/ield $M(\Psi)$ on~${\rm SO}(n;\mathbb R)$ is invariant with respect to ${\rm O}(i_1, \mathbb R)\times\dots\times {\rm O}(i_k, \mathbb R)$, i.e., $M(\Psi g)=M(\Psi)g$ for all $g\in {\rm SO}(n, \mathbb R)\bigcap ({\rm O}(i_1, \mathbb R)\times\dots\times {\rm O}(i_k,\mathbb R) )$. So we obtain a vector f\/ield $\widetilde M$ on the partial f\/lag space. In ef\/fect, all the conditions of Proposition~\ref{refgrad} hold (i.e., the function $F$ and the Riemannian structure~$\langle\,,\,\rangle_J$ are also invariant with respect to the group action); hence, the f\/ield~$\widetilde M$ is equal to the gradient of a~function~$\widetilde F$ with respect to the induced Riemannian structure.

\looseness=-1 Now we are going to show that the function $\widetilde F$ on ${\rm Fl}_{i_1,\dots,i_k}(\mathbb R)$ is Morse. To this end (see Proposition~\ref{morseref}) we should show that the restriction of the Hessian of $F$ on the directions orthogonal to the ``vertical'' vectors is nondegenerate. This Hessian of $F$ is given by formula~(39) in~\cite{CSS}: let $s_w\in {\rm SO}(n;\mathbb R)$ be a positive permutation matrix corresponding to $w\in S_n$ (i.e., the matrix in ${\rm SO}(n, \mathbb R)$ which permutes the vectors $\pm e_i$, $i=1,\dots,n$ so that the permutation $w\in S_n$ emerges when the signs are dropped); then $s_w$ is a singular point in ${\rm SO}(n;\mathbb R)$ and the Hessian of~$F$ at~$s_w$~is
\begin{gather}\label{eqhessian}
d^2_{s_w}F= \sum_{i<j} \theta_{ij}^{2}(\lambda_{w(i)}-\lambda_{w(j)})(j-i).
\end{gather}
Here $\theta_{ij}$ denote the local coordinates on ${\rm SO}(n, \mathbb R)$ at the point $s_w\in {\rm SO}(n, \mathbb R)$ obtained from the standard coordinates on $\mathfrak{so}_n$ by right translation. Put
\begin{gather*}
G={\rm SO}(n, \mathbb R)\bigcap\big({\rm O}(i_1, \mathbb R)\times\dots\times {\rm O}(i_k,\mathbb R)\big).
\end{gather*}
Recall that we assume the eigenvalues of $\Lambda$ to be partitioned into $k$ ``blocks'', see \eqref{eq:eige}, so that everything is $G$-invariant. In this case the ``vertical'' directions at $s_w$ are equal to the tangent space of the corresponding orbit
\begin{gather*}
T^v_{s_w}{\rm SO}(n;\mathbb R)=T_{s_w} (s_w\cdot G ).
\end{gather*}
The Lie algebra $\mathfrak g$ of $G$ is equal to the space of all antisymmetric matrices, spanned by the following set of elementary antisymmetric matrices:
\begin{gather*}
e_{ij}-e_{ji},\qquad \text{such that}\quad \exists\, p,\quad 1\le p\le k,\quad i_0+\dots+i_{p-1}< i<j\le i_0+\dots +i_p,
\end{gather*}
where we have put $i_0=0$ and use symbols $e_{ij}$ to denote the matrix units. Then the vertical directions at~$s_w$ are equal to the linear span of
\begin{gather*}
(e_{w(i)w(j)}-e_{w(j)w(i)})s_w\in T_{s_w}{\rm SO}(n;\mathbb R)
\end{gather*}
for \looseness=-1 the same set of indices $i$, $j$, as above. This follows either from the direct computations, or from the fact that the conjugate action $\operatorname{Ad}_{s_w}$ of~$s_w$ on $\mathfrak{so}_n$ amounts to the permutation of indices. Now the non-degeneracy of $d^2_{s_w}$ on the linear complement of $T^v_{s_w}{\rm SO}(n, \mathbb R)\subset T_{s_w}{\rm SO}(n, \mathbb R)$ can be seen from the comparison of formula~\eqref{eqhessian} with this description of the vertical subspace. Due to the $G$-invariance of the constructions, the same is true for an arbitrary point in the orbit through~$s_w$.

We sum up our observations in the following proposition:
\begin{predl}\label{morsepartflag}
The full symmetric Toda system on the matrices with non-distinct eigenvalues induces a dynamical system on the partial flag manifold ${\rm Fl}_{i_1,\dots,i_k}(\mathbb R)$ $($where $i_1,\dots,i_k$ are the multiplicities of eigenvalues$)$. This dynamical system is equal to the gradient flow of a~Morse function on ${\rm Fl}_{i_1,\dots,i_k}(\mathbb R)$.
\end{predl}

We conclude this section by an important observation: consider the natural projection (see Section~\ref{ss:bru:fs&gr})
\begin{gather*}
\pi\colon \ {\rm Fl}_n(\mathbb R)\to {\rm Fl}_{i_1,\dots,i_k}(\mathbb R).
\end{gather*}
As we have shown above, when the eigenvalues of $\Lambda$ are given by~\eqref{eq:eige}, the spaces on both sides of this diagram can be equipped with a Morse--Bott function and a Riemannian metric, which give rise to the gradient vector f\/ields $\widetilde M$ on them. In both cases these structures are pulled to the f\/lag spaces (full or partial) from the group ${\rm SO}(n, \mathbb R)$. On the other hand we can use this projection~$\pi$ to pull the structures we need from ${\rm Fl}_n(\mathbb R)$ to the partial f\/lag space. It is clear that in either way we shall obtain the same result. In other words, when the eigenvalues of $\Lambda$ are not distinct, the function $F$ on ${\rm SO}(n, \mathbb R)$ is a Morse--Bott function, which induces a Morse function on the base ${\rm Fl}_{i_1,\dots,i_k}(\mathbb R)$.

\section{The asymptotic behavior}\label{s:fin}
In this section we prove the main theorem of the paper: the asymptotic behavior of the trajectories of the vector f\/ield induced on ${\rm Fl}_{i_1,\dots,i_k}(\mathbb R)$ by the FS Toda lattice on the set of symmetric matrices with multiple eigenvalues (with multiplicities $i_1,\dots,i_k$) is completely determined by the Bruhat order on $S_n^I$, see Theorem~\ref{prop:traject1}. We begin with two particular cases in which the structure of the trajectories is easy to perceive: the case of real projective spaces and the case of Grassmannian ${\rm Gr}_{2,4}(\mathbb R)$. The general case is treated at the end.

\subsection{Example 1: projective spaces}\label{ss:fin:ex}
Let $L$ be the Lax matrix of the FS Toda lattice; $\Lambda$ is the diagonal matrix of its eigenvalues so that
\begin{gather*}
L=\Psi\Lambda\Psi^{-1},\qquad \Psi\in {\rm SO}(n,\mathbb R).
\end{gather*}
As we have explained this relation allows one to introduce an analogue of the FS Toda lattice on the partial f\/lag space ${\rm Fl}_{i_1,\dots,i_k}(\mathbb R)$, where $i_1,i_2,\dots,i_k$ are the multiplicities of the eigenvalues of~$L$. The geometry of the f\/lag space (and hence the structure of Toda trajectories) depends to a great measure on the partition $I=(i_1,i_2,\dots,i_k)$. In this section we consider the simplest possible case; namely, we assume that $I=(1,n-1)$, i.e., we assume that there are only two dif\/ferent eigenvalues $\lambda<\mu$ with multiplicities $1$ and $n-1$, respectively (we can also assume that $\lambda+(n-1)\mu=0$). The f\/lag spaces here are equal to the projective spaces~$\mathbb RP^{n-1}$.

We begin with the smallest possible dimension $n=3$ and eigenvalues $\lambda_1=\lambda<\lambda_2=\lambda_3=\mu$. In this case we consider the system on the projective plane $\mathbb RP^2$. One can see that there are exactly $3$ singular points of the vector f\/ield induced by the Toda system on $\mathbb RP^2$ corresponding to the permutations $0=(\lambda,\mu,\mu)$, $1=(\mu,\lambda,\mu)$ and $2=(\mu,\mu,\lambda)$; direct calculations in the local coordinates at these points (it is enough to transfer the coordinates from the unit of ${\rm SO}(3,\mathbb R)$ and take the directions complementary to the vertical) show that the Morse function takes distinct values at these points so that their Morse indices are~$2$,~$1$ and~$0$, respectively (see formula~\eqref{eqhessian}). Recall that the index of a singular point of a Morse function is equal to the dimension of the submanifold spanned by the trajectories exiting this point. Taking this into consideration we see that there is a unique way of connecting~$0$,~$1$ and $2$ by trajectories:
\begin{gather*}
\xymatrix{
0\ar[r] \ar@/^1pc/[rr] &1\ar[r] & 2.
}
\end{gather*}
It is easy to describe the structure of the vector f\/ield here: its pullback from $\mathbb RP^2$ to $S^2$ has 6 singular points, which can be identif\/ied with the intersections of the sphere with the coordinate axes; the indices of these points are $0,1$ and $2$ so that the opposite points have the same index. Besides this the vector f\/ield is invariant with respect to the natural involution of the sphere (exchange of the opposite points). We shall call this vector f\/ield \textit{the Toda field} on $\mathbb RP^2$, and use the same name for similar f\/ields on arbitrary projective spaces.

In the general case $n>3$ we shall have an analogous picture: just remark that the diagonal embedding of ${\rm SO}(n,\mathbb R)$ into ${\rm SO}(n+1,\mathbb R)$, given by
\begin{gather*}
\Psi\mapsto\begin{pmatrix}\Psi & 0\\ 0 & 1\end{pmatrix},
\end{gather*}
corresponds on the one hand, to the hyperplane embedding of projective spaces $\mathbb RP^{n-1}\to\mathbb RP^n$, and on the other hand, it is induced by the evident embedding of symmetric $n\times n$ matrices into the space of the symmetric $(n+1)\times (n+1)$ matrices:
\begin{gather*}
L\to\begin{pmatrix}L &0\\ 0 & \mu\end{pmatrix}.
\end{gather*}
This embedding sends the symmetric matrices with the spectrum
\begin{gather*}
\lambda<\underbrace{\mu=\mu=\dots=\mu}_{n-1\ \text{times}}
\end{gather*}
to the symmetric matrices with the spectrum
\begin{gather*}
\lambda<\underbrace{\mu=\mu=\dots=\mu}_{n\ \text{times}}.
\end{gather*}
Also, \looseness=-1 one can see that this inclusion intertwines Toda systems on bigger matrices and smaller matrices. Thus, the Toda vector f\/ield on $\mathbb RP^{n-1}$ coincides with the restriction of the Toda f\/ield from~$\mathbb RP^n$, and the latter f\/ield has one more singular point in the complement of~$\mathbb RP^{n-1}$ in~$\mathbb RP^n$. This new point has the maximum possible index. So we see that in the case of a generic projective space~$\mathbb RP^n$ the phase diagram of singular points and trajectories between them is as follows

\vspace{3mm}

\begin{gather*}
\xymatrix{
0\ar[r]\ar@/^1pc/[rr]\ar@/^3pc/[rrrr] & 1\ar[r]\ar@/^2pc/[rrr] & 2\ar[r]\ar@/^1pc/[rr] & \dots\ar[r] & n\,.
}
\end{gather*}
It is also possible to rephrase the results of this section in terms of the asymptotic behavior of the symmetric Lax matrix $L$: as one can see, when $t\to\pm\infty$ the matrix~$L$ tends to a diagonal matrix with eigenvalues~$\lambda$ and $\mu$ of multiplicities~$1$ and~$n-1$, respectively, in which $\lambda$ stands on an arbitrary position (depending on the corresponding multiset permutation).

\subsection[Example 2: ${\rm Gr}_{2,4}(\mathbb R)$]{Example 2: $\boldsymbol{{\rm Gr}_{2,4}(\mathbb R)}$}\label{ss:fin:gr24}
The next simplest case is when there are only two distinct eigenvalues of $L$, but the dimensions of the corresponding eigenspaces are greater than $1$. In this case the phase space of the system can be identif\/ied (see Section~\ref{ss:bru:fs&gr}) with the Grassmann space of all dimension $d>1$ hyperplanes in an $n>d+1$ dimensional Euclidian space~$\mathbb R^n$. This case is already quite complicated, so we restrict our discussion to the least-dimensional case: $n=4$, $d=2$. We assume that the eigenvalues of the $4\times 4$ symmetric Lax matrix $L$ are $\lambda_1=\lambda_2=\lambda$, $\lambda_3=\lambda_4=\mu$ and consider the induced gradient system on~${\rm Gr}_{2,4}(\mathbb R)$.

We begin with a detailed description of the Grassmannian. This is the manifold paramete\-ri\-zing all $2$-dimensional subspaces in~$\mathbb R^4$. As we mentioned in Section~\ref{ss:bru:fs&gr} one has a homeomorphism of spaces
\begin{gather*}
{\rm Gr}_{2,4}(\mathbb R)={\rm SO}(4,\mathbb R)/{\rm SO}(4,\mathbb R)\cap({\rm O}(2,\mathbb R)\times {\rm O}(2,\mathbb R)).
\end{gather*}
The identif\/ication is given by choosing an orthogonal basis in $\mathbb{R}^{4}$:
\begin{gather*} %{basis}
 e_{1}=\left(
\begin{matrix}
 1\\
 0\\
 0\\
 0
\end{matrix}
\right), \qquad
 e_{2}=\left(
\begin{matrix}
 0\\
 1\\
 0\\
 0
\end{matrix}
\right), \qquad
 e_{3}=\left(
\begin{matrix}
 0\\
 0\\
 1\\
 0
\end{matrix}
\right)
, \qquad
 e_{4}=\left(
\begin{matrix}
 0\\
 0\\
 0\\
 1
\end{matrix}
\right).
\end{gather*}
We \looseness=-1 can regard the $2$-dimensional plane spanned by $e_1$, $e_2$ as the ``origin'' $x$ of ${\rm Gr}_{2,4}(\mathbb R)$ so that all the other planes in $\mathbb R^4$ are equal to translations of $x$ by appropriate elements of ${\rm SO}(4,\mathbb R)$. The stabilizer of $x$ in ${\rm SO}(4,\mathbb R)$ is the subgroup $B_x$ of ${\rm SO}(4,\mathbb R)$ comprised of the elements of the form
\begin{gather*}%{Stab-x-Gr}
\left(
\begin{matrix}
 \cos (t_{a}) & \sin (t_{a}) & 0 & 0\\
 -\sin (t_{a}) & \cos (t_{a}) & 0 & 0\\
 0 & 0 & \cos (t_{b}) & \sin (t_{b})\\
 0 & 0 & -\sin (t_{b}) & \cos (t_{b})\\
\end{matrix}
\right),\qquad \left(
\begin{matrix}
 -\cos (t_{a}) & \sin (t_{a}) & 0 & 0\\
 \sin (t_{a}) & \cos (t_{a}) & 0 & 0\\
 0 & 0 & -\cos (t_{b}) & \sin (t_{b})\\
 0 & 0 & \sin (t_{b}) & \cos (t_{b})
\end{matrix}
\right).
\end{gather*}
As one can easily see, the plane $x$ is preserved by the action of these elements:
\begin{gather*}%{tranns-stab-x}
e_{1}B_{x}=e_{1}\cos (t_{a})+ e_{2}(-\sin (t_{a})),\qquad
e_{2}B_{x}=e_{1}\sin (t_{a})+ e_{2}\cos (t_{a}).
\end{gather*}
Now consider the generic point $y\in {\rm Gr}_{2,4}(\mathbb R)$: as we know $y=g^yx$ for some $g^y\in {\rm SO}(4,\mathbb R)$. Then the stabilizer of~$y$ is equal to the conjugation of $B_x$ by~$g^y$:
\begin{gather*}
B_y=g^yB_x(g^y)^{-1}.
\end{gather*}
This simple observation allows one to f\/ind a suitable description of the tangent space of~${\rm Gr}_{2,4}(\mathbb R)$ at~$y$:
\begin{gather*}
T_y{\rm Gr}_{2,4}(\mathbb R)\cong\mathfrak{so}(4)/\operatorname{Ad}_{g^y}(\mathfrak g),
\end{gather*}
where $\mathfrak g$ is the Lie algebra of the group $B_x$ and $\operatorname{Ad}_g$, $g\in {\rm SO}(4,\mathbb R)$ denotes the adjoint action of the group ${\rm SO}(4,\mathbb R)$ on its Lie algebra $\mathfrak{so}(4)$ (by conjugation of matrices).

We choose the standard basis in $\mathfrak{so}(4)$ (the space of all $4\times 4$ anti-symmetric matrices) so that the generic element~$X$ of $\mathfrak{so}(4)$ takes the form
\begin{gather*}
X=\begin{pmatrix}
0 & \theta_1 & \theta_2 & \theta_3\\
-\theta_1& 0 & \theta_4 & \theta_5\\
-\theta_2 & -\theta_4 & 0 & \theta_6\\
-\theta_3 & -\theta_5 & -\theta_6 & 0
\end{pmatrix}.
\end{gather*}
The functions $\theta_1,\dots,\theta_6$ are coordinates in $\mathfrak{so}(4)$. The exponential map allows one to pull these coordinates to an open neighbourhood of the unit matrix in~${\rm SO}(4,\mathbb R)$ and the right translations by the group elements then give coordinate systems in open neighbourhoods of the points in~${\rm SO}(4,\mathbb R)$ and on tangent spaces at these points. Below we shall use these coordinates extensively without explanation, preserving (by a slight abuse of the language) their original names.

This notation is well-suited for the description of the tangent spaces of ${\rm Gr}_{2,4}(\mathbb R)$ at the singular points of the FS Toda lattice: recall that these singular points are equal to the projections to~${\rm Gr}_{2,4}(\mathbb R)$ of the singular f\/ibres of the FS Toda lattice in~${\rm SO}(4,\mathbb R)$ (or in ${\rm Fl}_{4}(\mathbb{R})$). These singular f\/ibres are determined by the singular points of the FS Toda lattice on the matrices with distinct eigenvalues, which they contain. As one knows, the singular points in ${\rm SO}(4,\mathbb R)$ of the FS Toda lattice corresponding to the Lax matrices with distinct eigenvalues are given by the matrices, which permute the basis vectors and their opposites (see~\cite{CSS}). The conjugation by such matrices induces permutations of the coordinates $\theta_1,\dots,\theta_6$.

With the help of these observations one can draw the following table, describing the tangent spaces at the critical points, values of the Morse function $F_{2,4}$, positive and negative directions of Hessians at these critical points, i.e., of the singular points of the Toda vector f\/ield on ${\rm Gr}_{2,4}(\mathbb R)$. We index the singular points on ${\rm Gr}_{2,4}(\mathbb R)$ by the corresponding permutations of the eigenvalues' set $(\lambda,\lambda,\mu,\mu)$, $\lambda<\mu$.
$$
\begin{tabular}{|c|c|c|c|c|c|}
\hline\cline{1-0}
\text{Point} & \text{Morse index} & \text{Value of $F_{2,4}$} & $+$ & $-$ & \text{Minors}\tsep{2pt}\bsep{2pt}\\
\hline\cline{1-0}
$(\lambda, \lambda, \mu, \mu)$ & 0 & $\lambda + 5 \mu$ & $\theta_{2}$, $\theta_{3}$, $\theta_{4}$, $\theta_{5}$ & 0 & $\psi_{13}$, $\psi_{14}$, $\psi_{41}$, $\psi_{42}$\tsep{2pt}\bsep{2pt}\\
\hline\cline{1-0}
$(\lambda, \mu, \lambda, \mu)$ & 1 & $2\lambda + 4 \mu$ & $\theta_{1}$, $\theta_{4}$, $\theta_{6}$ & $\theta_{2}$ & $\psi_{13}$, $\psi_{14}$, $\psi_{41}$, $\psi_{42}$\tsep{2pt}\bsep{2pt}\\
\hline\cline{1-0}
$(\lambda, \mu, \mu, \lambda)$ & 2 & $3\lambda + 3 \mu$ & $\theta_{5}$, $\theta_{6}$ & $\theta_{1}$, $\theta_{3}$ & $\psi_{13}$, $\psi_{14}$, $\psi_{43}$, $\psi_{44}$\tsep{2pt}\bsep{2pt}\\
\hline\cline{1-0}
$(\mu, \lambda, \lambda, \mu)$ & 2 & $3\lambda + 3 \mu$ & $\theta_{1}$, $\theta_{3}$ & $\theta_{5}$, $\theta_{6}$ & $\psi_{11}$, $\psi_{12}$, $\psi_{41}$, $\psi_{42}$\tsep{2pt}\bsep{2pt}\\
\hline\cline{1-0}
$(\mu, \lambda, \mu, \lambda)$ & 3 & $4\lambda + 2 \mu$ & $\theta_{2}$ & $\theta_{1}$, $\theta_{4}$, $\theta_{6}$ & $\psi_{11}$, $\psi_{12}$, $\psi_{43}$, $\psi_{44}$\tsep{2pt}\bsep{2pt}\\
\hline\cline{1-0}
$(\mu, \mu, \lambda, \lambda)$ & 4 & $5\lambda + \mu$ & 0 & $\theta_{2}$, $\theta_{3}$, $\theta_{4}$, $\theta_{5}$ & $\psi_{11}$, $\psi_{12}$, $\psi_{43}$, $\psi_{44}$\tsep{2pt}\bsep{2pt}\\
\hline
\end{tabular}
$$
In the \looseness=-1 fourth and the f\/ifth columns we use the coordinates $\theta_1,\dots,\theta_6$ on ${\rm SO}(4,\mathbb R)$ introduced earlier and project them to ${\rm Gr}_{2,4}(\mathbb R)$: on the level of tangent spaces the dif\/ferential $d\pi$ of this projection is a linear epimorphism and we identify $T_{\pi(x)}{\rm Gr}_{2,4}(\mathbb R)$ with a linear subspace, which maps isomorphically onto it. The last column gives the list of certain invariant surfaces of the Toda f\/ield in ${\rm SO}(4,\mathbb R)$, which contain the corresponding ``singular f\/ibre'' of the Toda f\/ield (i.e., the f\/ibre of the projection $\pi\colon {\rm SO}(4,\mathbb R)\to {\rm Gr}_{2,4}(\mathbb R)$, on which the Toda f\/ield vanishes identically). As one knows (see~\cite{CSS}) such surfaces can be given by the equations $\psi_{ij}=0$; there are more general surfaces of this sort, in which the equations are given by the condition that a~certain minor of the matrix $\Psi=(\psi_{ij})$ vanishes (hence the name ``minor surfaces'' that we use in this and previous paper) but we restrict our attention to the simplest set of invariants listed in this table.

We illustrate these constructions by Fig.~\ref{bundle-1}, in which the case of the singular point in ${\rm Gr}_2(4,\mathbb R)$, corresponding to the permutation $(\lambda, \mu, \lambda, \mu)$, is considered (below we shall often identify such points with the corresponding permutations of eigenvalues).
\begin{figure}[t]\centering
\includegraphics[width=192pt]{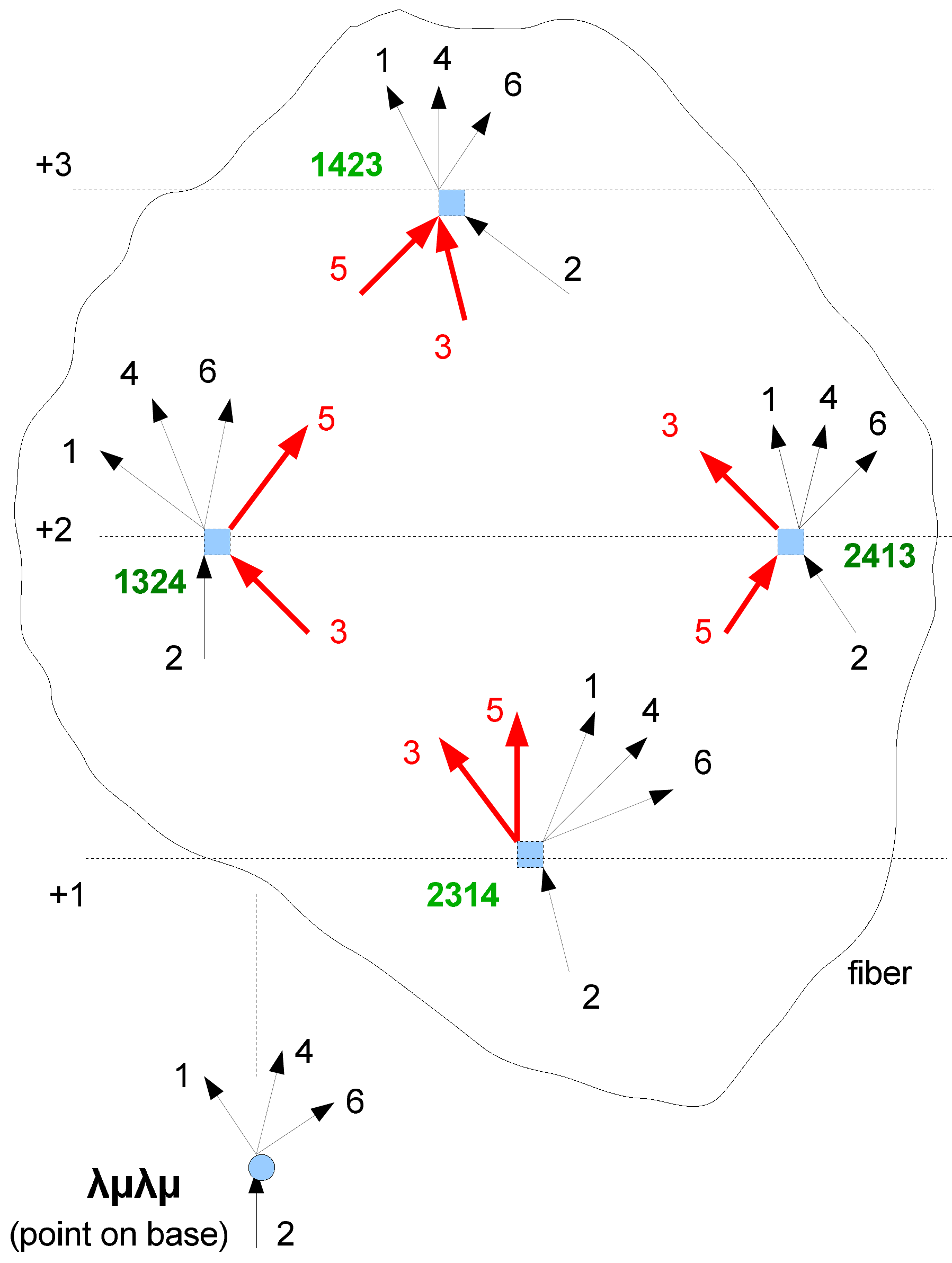}
\caption{The bundle over $(\lambda, \mu, \lambda, \mu)\in {\rm Gr}_{2,4}(\mathbb R)$.}\label{bundle-1}
\end{figure}

\looseness=-1 In this f\/igure the little circle at the bottom represents the point in ${\rm Gr}_{2,4}(\mathbb R)$ that corresponds to the permutation $(\lambda,\mu,\lambda,\mu)$ of the eigenvalues. The small squares above correspond to the permutation matrices in ${\rm SO}(4,\mathbb R)$ that project into the chosen point. The arrows represent the coordinate directions in the tangent space of the group at the chosen points. The red (bold) arrows correspond to the directions inside the f\/ibre, and the black (thin) arrows to all the rest. As one knows (see~\cite{CSS}) the coordinates~$\theta_i$ are (up to inf\/initesimal correction terms) the canonical coordinates, in which the Morse function takes the form of the sums of the squares. So the arrows are directed towards or from the points, depending on whether the corresponding tangent directions are in positive or negative subspaces of the Hessian (the red arrows in fact correspond to the directions which lie in the kernel of the Hessian but we preserve them for the sake of simplicity; their directions are determined under the assumption that all the eigenvalues are distinct and are ordered in a natural way, otherwise one can take arbitrary vectors in the f\/ibre direction).

Finally, we consider the minor surfaces that pass through the f\/ibre: it is clear that if a~surface is invariant under the Toda f\/low on ${\rm SO}(4,\mathbb R)$, its projection to the Grassmann space is an invariant set of the generalized f\/lows. Using this simple observation we come up with the following diagram, see Fig.~\ref{bundle-2}, representing the f\/lows in ${\rm Gr}_{2,4}(\mathbb R)$; the dotted lines represent the $1$-parameter families of trajectories that connect points whose Morse indices dif\/fer by $2$. There are also $2$-parameter families between the points with the Morse indices that dif\/fer by $3$ and one $3$-parametric family between the lowest and the highest points, which we have omitted to make the diagram more readable.

\begin{figure}[t]
$$
\xymatrix{
{}\ar@<-1.5ex>@{.}[rrrrr] & {\mathrm{ind}=4} & & {\mu\mu\lambda\lambda} & & \\
{}\ar@<-1.5ex>@{.}[rrrrr] & {\mathrm{ind}=3} & & {\mu\lambda\mu\lambda}\ar[u] & &\\
{}\ar@<-1.5ex>@{.}[rrrrr] & {\mathrm{ind}=2} & {\mu\lambda\lambda\mu}\ar[ur]\ar@{-->}[uur]& &{\lambda\mu\mu\lambda}\ar[ul]\ar@{-->}[uul] &\\
{}\ar@<-1.5ex>@{.}[rrrrr] & {\mathrm{ind}=1} & &{\lambda\mu\lambda\mu}\ar[ul]\ar[ur]\ar@{-->}[uu] & & \\
{}\ar@<-1.5ex>@{.}[rrrrr] & {\mathrm{ind}=0} & & {\lambda\lambda\mu\mu}\ar[u]\ar@{-->}[uul]\ar@{-->}[uur] & &
}
$$
\caption{The generalized Toda f\/low on ${\rm Gr}_2(4,\mathbb R)$.}\label{bundle-2}
\end{figure}
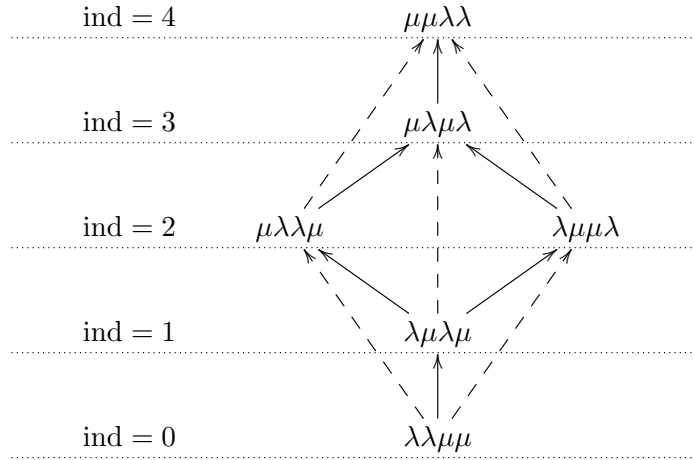

As one can see, this diagram coincides with the Hasse diagram of the Bruhat order on multiset permutations. Below (see Theorem~\ref{prop:traject1}), we shall show that it is always the case.

Let us give a brief explanation of how the diagram has been obtained: consider two singular points in ${\rm Gr}_{2,4}(\mathbb{R})$, which correspond to the permutations $(\lambda, \lambda, \mu, \mu)$ and $(\lambda, \mu, \lambda, \mu)$ of the eigenvalues. Let us show that there is a f\/inite number of trajectories of the Toda f\/low between these points.

Observe that the surface
\begin{gather*}
\Sigma = (\psi _{13}=0) \cap (\psi _{14}=0) \cap (\psi _{41}=0) \cap (\psi _{42}=0),\qquad \Sigma\subset {\rm SO}(4,\mathbb R)
\end{gather*}
is invariant with respect to the ${\rm O}(2,\mathbb R)\times {\rm O}(2,\mathbb R)$-action and also is preserved by the Toda f\/low. Thus, it projects to an invariant surface $\widetilde\Sigma$ in ${\rm Gr}_{2,4}(\mathbb R)$, which contains only the above mentioned singular points on the base. Of course, $\Sigma$ and $\widetilde\Sigma$ are singular varieties in ${\rm SO}(4,\mathbb R)$ and ${\rm Gr}_{2,4}(\mathbb R)$, so for our purposes it is enough to compute their dimensions at the generic point. The generic matrix $\Psi$ from $\Sigma$ has the following form
\begin{gather*} %{PsiSigma}
\Psi = \left(
\begin{matrix}
 \psi _{11} & \psi _{12} & 0 & 0 \\
 \psi _{21} & \psi _{22} & \psi _{23} & \psi _{24} \\
 \psi _{31} & \psi _{32} & \psi _{33} & \psi _{34} \\
 0 & 0 & \psi _{43} & \psi _{44}
\end{matrix}
\right).
\end{gather*}
Moreover, $\Psi$ being an orthogonal matrix the dimension of the surface~$\Sigma$ at the generic point is equal to~$3$: it is enough to compute the dimension of the tangent space to it at the unit matrix; the tangent space of~${\rm SO}(4,\mathbb R)$ at the unit consists of antisymmetric matrices, so summing the conditions on the matrix elements we come to the above mentioned conclusion about the dimension of~$\Sigma$.

Thus, the dimension $\widetilde\Sigma$ (at the generic point) is equal to $1$. On the other hand $\widetilde\Sigma$ is inva\-riant with respect to the Toda f\/ield on ${\rm Gr}_{2,4}(\mathbb R)$; since it is $1$-dimensional, it should consist of a~f\/inite number of trajectories connecting singular points in it. Thus, the points $(\lambda, \lambda, \mu, \mu)$ and $(\lambda, \mu, \lambda, \mu)$ must be connected by a discrete set of trajectories.

\subsection{The general case}\label{ss:fin:gc}
Finally, let us consider the most general distribution of the eigenvalues of $L$; let us f\/ix the multi-index $I=(i_1,i_2,\dots,i_k)$ such that $0<i_j$, $i_1+\dots+i_k=n$. As we have observed earlier, the set of all real symmetric matrices $L$ with the given set $\Lambda$ of eigenvalues $\lambda_1<\lambda_2<\dots<\lambda_k$, such that $\lambda_j$ has multiplicity~$i_j$, is homeomorphic to the f\/lag space ${\rm Fl}_I(\mathbb R)$ (see Section~\ref{ss:bru:fs&gr}). We use the same symbol $\Lambda$ for a diagonal matrix with the set of eigenvalues equal to $\Lambda$.

As we have shown earlier (see Section~\ref{ss:tod:pfl}) the FS Toda lattice on the space of symmetric matrices conjugate with~$\Lambda$ is induced by a~gradient vector f\/ield $\xi_I$ on the partial f\/lag space~${\rm Fl}_I(\mathbb R)$: the image of the usual Toda vector f\/ield~$\xi$ on the full f\/lag variety ${\rm Fl}_n(\mathbb R)$. Moreover, the potential function generating this f\/ield is a~Morse function; its singular points correspond to those f\/ibres of the natural projection ${\rm Fl}_n(\mathbb R)\to {\rm Fl}_I(\mathbb R)$ on which the Toda vector f\/ield vanishes identically (singular f\/ibres with respect to the Toda potential function on ${\rm Fl}_n(\mathbb R)$). As one knows, in this case the trajectories connect the singular points of the f\/ield when $t\to\pm\infty$. Our goal is to describe the order in which these points are connected: we shall show that this order is the same as the (strong) Bruhat order.

To this end recall (see Section~\ref{ss:tod:pfl} and~\cite{CSS}) that the singular points of the Toda vector f\/ield on the full f\/lag space~${\rm Fl}_n(\mathbb R)$ with the Lax matrix, whose eigenvalues are all distinct, correspond to the permutation matrices in~${\rm SO}(n,\mathbb R)$, i.e., matrices which permute the basis vectors in~$\mathbb R^n$ and if necessary, multiply them by~$-1$ (to make sure the determinant is positive). In order to understand the structure of the singularities in ${\rm Fl}_I(\mathbb R)$, it is convenient to look at the bundle
\begin{gather*}
{\rm SO}(n,\mathbb R)\to {\rm Fl}_I(\mathbb R)
\end{gather*}
with the f\/ibres isomorphic to $G={\rm SO}(n,\mathbb R)\bigcap({\rm O}(i_1,\mathbb R)\times\dots\times {\rm O}(i_k,\mathbb R))$. The vector f\/ield $\xi$ can be raised to a vector f\/ield on the orthogonal group with singular points at the matrices~$A_w$. The f\/ibre through a permutation matrix~$A_w$, which is equal to the conjugation of~$G$ by~$A_w$, will contain all permutations of the form~$wuw^{-1}$, where $u\in S_{i_1}\times\dots\times S_{i_k}$. Comparing this description with Section~\ref{ss:bru:p&pfl} we conclude that the following proposition holds:
\begin{predl}\label{prop:sing}
The singular points of the Morse field induced on ${\rm Fl}_I(\mathbb R)$ by the FS Toda lattice with eigenvalues $\Lambda$ $($of multiplicities~$I)$ are indexed by the multiset permutations~$S_n^I$.
\end{predl}

{\sloppy
Now it is our purpose to describe the trajectories connecting dif\/ferent singular points in~${\rm Fl}_I (\mathbb R)$. We shall prove the following statement:

}

\begin{theor}\label{prop:traject1}
Let $\psi, \omega\in S_n^I$ be two multiset permutations. Then the corresponding singular points of the FS Toda lattice in ${\rm Fl}_I(\mathbb R)$ will be connected by a trajectory, if and only if $\psi\prec\omega$ in the Bruhat order on $S_n^I$ $($see Section~{\rm \ref{ss:bru:p&pfl})}. Moreover, the dimension of the subset swept by the trajectories connecting these points is equal to the length of the path in the Hasse diagram of~$S_n^I$.
\end{theor}

\begin{proof}First of all observe that the projection ${\rm Fl}_n(\mathbb R)\to {\rm Fl}_I(\mathbb R)$ maps the Toda f\/ield $\xi$ to the f\/ield $\xi_I$; hence it sends invariant subvarieties in the full f\/lag space to invariant subsets in ${\rm Fl}_I(\mathbb R)$. In particular, it means that the Schubert cells in ${\rm Fl}_I(\mathbb R)$ are preserved by the generalized Toda f\/low since the Schubert cells in the full f\/lags are. Moreover, since the Schubert cells in ${\rm Fl}_n(\mathbb R)$ coincide with the \textit{unstable} subspace of~$\xi$ (i.e., the possibly singular subspace in $M$ swept by the trajectories tending to the given singular point of the gradient system) we conclude that their images in the partial f\/lags coincide with the stable subspaces of $\xi_I$: this follows for example from the fact that the cells in ${\rm Fl}_I(\mathbb R)$ are the homeomorphic images of \textit{the minimal} cells in~${\rm Fl}_n(\mathbb R)$, see the end of Section~\ref{ss:bru:schucells} (also compare formula~\eqref{eqhessian}).

Similarly, the \textit{stable} submanifolds of $\xi_I$ (i.e., the subsets spanned by the outgoing trajectories of~$\xi_I$) in~${\rm Fl}_I(\mathbb R)$ coincide with the dual Schubert cells in this space since this is true for the dual cells in~${\rm Fl}_n(\mathbb R)$ (this was proved in~\cite{CSS}). However, we know (see Section~\ref{ss:bru:schucells} again) that the Schubert cell and dual Schubert cell in the f\/lag space intersect if and only if the corresponding multiset permutations are comparable in the Bruhat order. The statement about the dimensions follows from the transversality of these intersections.
\end{proof}

We conclude by the simple observation concerning the picture that emerges on the level of the Lax matrices:
\begin{cor}\label{cor:fin}
The Toda flow $t\to L(t)$ on symmetric matrices converges to a diagonal matrix when $t\to\pm\infty$ $($the set of eigenvalues of these matrices is fixed$)$. Two such matrices are connected by a trajectory if and only if the corresponding permutations of the eigenvalues are comparable with respect to the Bruhat order on permutations $($or permutations with repetitions if there are multiple eigenvalues$)$.
\end{cor}

\section{Further observations}
\looseness=-1 As we have shown above (see Section~\ref{s:tod}), the FS Toda lattice on symmetric matrices with non-distinct eigenvalues can be regarded as a dynamical system on partial f\/lag manifolds. It is interesting that although the system on these spaces can be described as the image of the FS Toda lattice, the usual invariants of the FS Toda system do not descend easily to the f\/lag spaces: direct computations show that they can become constants or functionally-dependent. So the question is whether one can still f\/ind new invariants to prove Liuoville integrability of such systems (of course, one should f\/irst make up explicit def\/initions of the Poisson structures used there).

One can begin with studying a few low-dimensional cases, f\/irst of all those which correspond to the projective spaces (see Section~\ref{ss:fin:ex}). Already in the simplest case $n=3$ (and when two eigenvalues coincide), that is for the system on $\mathbb RP^2$, we obtain a new integral of motion:
\begin{gather*} %{N3-invariant-E-2}
I_{(RP^{2}, \lambda_{1}=\lambda_{2})}= \frac{1}{(\mu -\lambda)}\frac{\psi_{23}^{2}}{\psi_{13}^{2}\psi_{33}^{2}}.
\end{gather*}
Here $\lambda_1=\lambda_2=\lambda$ and $\lambda_3=\mu$, $\lambda<\mu$ are the eigenvalues. In terms of the matrix entries~$a_{ij}$ of the Lax matrix $L$ this function can be rewritten in the following form:
\begin{gather*} %{N3-invariant-E-1}
I_{(RP^{2}, \lambda_{1}=\lambda_{2})}(a_{ij})= \frac{a_{12}a_{23}}{a_{13}^{3}}.
\end{gather*}
On the other hand, one can show that in this case the chopping procedure (see~\cite{DLNT} for example) does not give any integrals which are dif\/ferent from the traces of the powers of the Lax matrix. Thus, this invariant is a new phenomenon, which makes the whole picture quite intriguing. Similar integrals can be found in the case of projective spaces in higher dimensions. These questions will be the subject of our further papers.

\subsection*{Acknowledgments}
The authors would like to thank G.~Koshevoy for the fruitful discussion. We also would like to thank the referees, whose remarks helped in a great measure to improve the paper. The work of Yu.B.~Chernyakov was supported by grant RFBR-15-01-08462. The work of G.I.~Sharygin was supported by grant RFBR-15-01-05990. The work of A.S.~Sorin was partially supported by RFBR grants 15-52-05022-Arm-a and 16-52-12012-NNIO-a.

\pdfbookmark[1]{References}{ref}
\LastPageEnding

\end{document}